\documentclass[11pt]{article}

\usepackage{amsmath,amssymb,amsfonts,amsthm}
\usepackage{moreverb,rotating,graphics}
\usepackage[T1]{fontenc}
\usepackage[latin1]{inputenc}
\usepackage{epsfig}
\usepackage[francais,english]{babel} 
\newtheorem{theorem}{Theorem}

\newtheorem{remark}{Remark}
\newtheorem{lemma}{Lemma}

\def\tr{{\rm Tr \;}}
\def\div{{\rm div \;}}
\def\NN{\mathbb{N}}
\def\ZZ{\mathbb{Z}}
\def\RR{\mathbb{R}}

\newcommand \dps{\displaystyle }

\newcommand \br {\bold r}
\newcommand{\gS}{\mathfrak{S}}

\newcommand \bR {{\bf R}}
\newcommand \bx {{\bf x}}

\newcommand \Rspin {\RR^3_\Sigma}

\newcommand \spinup {|\!\!\uparrow\rangle}
\newcommand \spindown {|\!\!\downarrow\rangle}

\title{On Kohn-Sham models with LDA and GGA
  exchange-correlation functionals} 
\author{Arnaud Anantharaman and Eric Canc\`es \\ \\
       \footnotesize{CERMICS, Ecole des Ponts and Université Paris Est} \\
       \footnotesize{6 \& 8 avenue Blaise Pascal, 77455 Marne-la-Vallée
         Cedex 2, France} \\
       \footnotesize{and} \\
       \footnotesize{INRIA Rocquencourt, Micmac Project Team} \\
       \footnotesize{Domaine de Voluceau,  B.P. 105,
 78153 Le Chesnay Cedex, France} \\ 
          }
\begin{document}

%\selectlanguage{francais}
\selectlanguage{english}

\maketitle

\begin{abstract}
This article is concerned with the mathematical analysis of the
Kohn-Sham and extended Kohn-Sham models, in the local density
approximation (LDA) and generalized gradient approximation (GGA)
frameworks. After recalling the mathematical derivation of the Kohn-Sham
and extended Kohn-Sham LDA and GGA models from the Schrödinger equation,
we prove that the extended Kohn-Sham LDA model has a solution for
neutral and positively charged systems. We then prove a similar result
for the spin-unpolarized Kohn-Sham GGA model for two-electron systems,
by means of a concentration-compactness argument. 
\end{abstract}

%{\bf Keywords:} Density functional theory, Kohn-Sham,
%concentration-compactness 

%\newpage
%\tableofcontents

\section{Introduction}

\noindent
Density Functional Theory (DFT) is a powerful, widely used method
for computing approximations of ground state electronic energies and
densities in chemistry, materials science, biology and
nanosciences.

According to DFT~\cite{HK,Lieb}, the electronic ground state energy
and density of a given molecular system can be obtained by solving a
minimization problem of the form
$$
\inf \left\{ F(\rho) + \int_{\RR^3} \rho V, \; \rho \ge 0, \;
  \sqrt{\rho} \in H^1(\RR^3), \; \int_{\RR^3} \rho = N \right\}
$$
where $N$ is the number of electrons in the system, $V$ the
electrostatic potential generated by the nuclei, and $F$ some
functional of the electronic density $\rho$, the functional $F$ being
universal, in the sense that it does not 
depend on the molecular system under consideration. Unfortunately, no
tractable expression for $F$ is known,
which could be used in numerical simulations.

The groundbreaking contribution which turned DFT into a useful tool
to perform calculations, is due to Kohn and Sham~\cite{KS}, who
introduced the local density approximation (LDA) to 
DFT. The resulting Kohn-Sham LDA model is still commonly used, in
particular in solid
state physics. Improvements of this model have
then been proposed by many authors, giving rise to Kohn-Sham GGA models
\cite{GGA1,GGA2,GGA3,PBE}, GGA being the abbreviation of Generalized Gradient
Approximation. While there is basically a unique Kohn-Sham LDA model,
there are several Kohn-Sham GGA models, corresponding to different
approximations of the so-called exchange-correlation functional. A given
GGA model will be known to perform well for some 
classes of molecular system, and poorly for some other classes. In some
cases, the best result will be obtained with LDA. It is to be noticed
that each Kohn-Sham model exists in two versions: the standard version,
with integer occupation numbers, and the extended version with
``fractional'' occupation numbers. As explained below, the former one
originates from Levy-Lieb's (pure state) contruction of the density
functional, while the latter is derived from Lieb's (mixed state)
construction.  

To our knowledge, there are very few results on Kohn-Sham LDA and GGA
models in the mathematical literature. In fact, we are only aware of a
proof of existence of a minimizer for the standard Kohn-Sham LDA model
by Le Bris~\cite{LeBris}. In this contribution, we prove the existence
of a minimizer for the extended Kohn-Sham LDA model, as well as for the
two-electron standard and extended Kohn-Sham GGA models, 
under some conditions on the GGA exchange-correlation functional. 

Our article is organized as follows. First, we
provide a detailed 
presentation of the various Kohn-Sham models, which, despite 
their importance in physics and chemistry~\cite{pubDFT}, are not very well
known in the mathematical community. The mathematical foundations of
DFT are recalled in section~\ref{sec:DFT}, and the derivation of the
(standard and extended) Kohn-Sham LDA and GGA models is discussed in
section~\ref{sec:KS}. We state our main results in
section~\ref{sec:main}, and postpone the proofs until
section~\ref{sec:proofs}. 

We restrict our mathematical analysis to closed-shell,
spin-unpolarized
models. All our results related to the LDA setting can be easily
extended to open-shell, 
spin-polarized models (i.e. to the local spin-density approximation
LSDA). Likewise, we only deal with all 
electron descriptions, but valence 
electron models with usual pseudo-potential approximations (norm
conserving~\cite{TM}, ultrasoft~\cite{ultrasoft}, PAW~\cite{PAW}) can
be dealt with in a similar way.

\section{Mathematical foundations of DFT}
\label{sec:DFT}

As mentioned previously, DFT aims at calculating electronic ground
state energies and densities. Recall that the ground state
electronic energy of a molecular system composed of $M$ nuclei of
charges $z_1$, ..., $z_M$ ($z_k \in \NN \setminus \left\{0\right\}$ in
atomic units) and $N$ electrons is the bottom of the spectrum
of the electronic hamiltonian 
\begin{equation} \label{eq:elec_ham}
H_N = - \frac 12 \sum_{i=1}^N \Delta_{\br_i} - \sum_{i=1}^N \sum_{k=1}^M
\frac{z_k}{|\br_i-\bR_k|} + \sum_{1 \le i < j \le N} \frac{1}{|\br_i-\br_j|}
\end{equation}
where $\br_i$ and $\bR_k$ are the positions in $\RR^3$ of the
$i^{\rm th}$ electron and the $k^{\rm th}$ nucleus respectively. The hamiltonian
$H_N$ acts on electronic wavefunctions
$\Psi(\br_1,\sigma_1;\cdots;\br_N,\sigma_N)$, $\sigma_i \in \Sigma :=
\left\{\spinup,\spindown\right\}$ denoting the spin
variable of the $i^{\rm th}$ electron, the nuclear coordinates
$\left\{\bR_k\right\}_{1 \le k \le M}$ playing the role of
parameters. It is convenient to denote by $\RR^3_\Sigma :=
\RR^3 \times \left\{\spinup,\spindown\right\}$ and $\bx_i
:= (\br_i,\sigma_i)$.  As electrons are fermions,
electonic wavefunctions are antisymmetric with respect to the
renumbering of electrons, i.e. 
$$
\Psi(\bx_{p(1)},\cdots,\bx_{p(N)}) = \epsilon(p)
\Psi(\bx_1,\cdots,\bx_N) 
$$
where $\epsilon(p)$ is the signature of the permutation $p$. Note that
(in the absence of magnetic fields) $H_N\Psi$ is real-valued if $\Psi$
is real-valued. Our purpose being the calculation of the bottom of the
spectrum of $H_N$, there is therefore no restriction in considering
real-valued wavefunctions only. In other words, $H_N$ can be considered
here as an operator on the real Hilbert space 
$$
{\cal H}_N = \bigwedge_{i=1}^N L^2(\RR^3_\Sigma),
$$
endowed with the inner product
$$
\langle \Psi | \Psi' \rangle_{{\cal H}_N} = \int_{(\RR^3_\Sigma)^N}
\Psi(\bx_1,\cdots,\bx_N) \, \Psi'(\bx_1,\cdots,\bx_N) \, d\bx_1
\cdots d\bx_N, 
$$
where 
$$
\int_{\RR^3_\Sigma} f(\bx) \, d\bx := \sum_{\sigma \in \Sigma} \int_{\RR^3}
f(\br,\sigma) \, d\br,
$$
and the corresponding norm $\| \cdot \|_{{\cal H}_N} = \langle \cdot |
\cdot \rangle_{{\cal H}_N}^{\frac 12}$.
It is well-known that $H_N$ is a self-adjoint operator on ${\cal H}_N$
with form domain 
$$
{\cal Q}_N = \bigwedge_{i=1}^N H^1(\RR^3_\Sigma).
$$
Denoting by
$$
Z = \sum_{k=1}^M z_k
$$
the total nuclear charge of the system, it results from Zhislin's theorem
that for neutral or positively charged systems ($Z \ge N$), $H_N$ has an
infinite number of negative eigenvalues below the bottom of its
essential spectrum. In particular, the electronic ground state energy
$\lambda_1(H_N)$ is an eigenvalue of $H_N$, and more precisely the
lowest one. 

In any case, i.e. whatever $Z$ and $N$, we always have
\begin{equation} \label{eq:GS}
\lambda_1(H_N) = \inf \left\{ \langle \Psi|H_N|\Psi\rangle, \; \Psi \in
  {\cal Q}_N, \; \|\Psi\|_{{\cal H}_N} =1 \right\}.
\end{equation}
Note that it also holds
\begin{equation} \label{eq:EGS}
\lambda_1(H_N) = \inf \left\{ \tr(H_N\Gamma), \; \Gamma \in {\cal
    S}({\cal H}_N), \; \mbox{Ran}(\Gamma) \subset {\cal Q}_N, \;
 0 \le \Gamma \le 1,  \; \tr(\Gamma)=1 \right\}.
\end{equation}
In the above expression, ${\cal S}({\cal H}_N)$ is the vector space of
bounded self-adjoint operators on ${\cal H}_N$, and the condition
$0 \le \Gamma \le 1$ stands for 
$0 \le \langle \Psi |\Gamma|\Psi \rangle \le \|\Psi\|_{{\cal H}_N}^2$
for all $\Psi \in 
{\cal H}_N$. Note that if $H$ is a bounded-from-below self-adjoint
operator on some Hilbert space ${\cal H}$, with form domain ${\cal Q}$, and
if $D$ is a positive trace-class self-adjoint operator on ${\cal H}$,
$\tr(HD)$ can always be defined in $\RR_+ \cup \left\{+\infty\right\}$
as {$ \tr(HD) = \tr((H-a)^{\frac 12}D(H-a)^{\frac 12})+a\tr(D)$} where $a$ is any real
number such that $H \ge a$.

 From a physical viewpoint, (\ref{eq:GS}) and (\ref{eq:EGS})
mean that the ground state energy can be computed either by minimizing
over pure states (characterized by wavefunctions $\Psi$) or by
minimizing over mixed states (characterized by density operators
$\Gamma$).

With any $N$-electron wavefunction $\Psi \in {\cal H}_N$ such that
$\|\Psi\|_{{\cal H}_N} =1$ can be associated the electronic density
$$
\rho_\Psi(\br) = N \; \sum_{\sigma \in \Sigma} \int_{(\RR^3_\Sigma)^{N-1}}
|\Psi(\br,\sigma;\bx_2,\cdots;\bx_N)|^2 \, d\bx_2 \cdots
d\bx_N.  
$$
Likewise, one can associate with any $N$-electron density operator
$\Gamma\in {\cal S}({\cal H}_N)$ such that $0 \le \Gamma \le 1$ and
$\tr(\Gamma)=1$, the electronic density
$$
\rho_\Gamma(\br) = N \; \sum_{\sigma \in \Sigma} \int_{(\RR^3_\Sigma)^{N-1}}
\Gamma(\br,\sigma;\bx_2,\cdots,\bx_N ; 
\br,\sigma;\bx_2,\cdots;\bx_N) \, d\bx_2\cdots
d\bx_N
$$
(here and below, we use the same notation for an operator and
its Green kernel).

Let us denote by 
$$
V(\br) = -\sum_{k=1}^M \frac{z_k}{|\br-\bR_k|}
$$
the electrostatic potential generated by the nuclei, and by
\begin{equation} \label{eq:interacting}
H_N^1 = 
- \frac 12 \sum_{i=1}^N \Delta_{\br_i} + \sum_{1 \le i < j \le N}
\frac{1}{|\br_i-\br_j|}. 
\end{equation}
It is easy to see that
$$
\langle \Psi|H_N|\Psi \rangle = \langle \Psi|H_N^1|\Psi \rangle +
\int_{\RR^3} \rho_\Psi V \quad \mbox{and} \quad
\tr(H_N\Gamma) = \tr(H_N^1\Gamma) + \int_{\RR^3} \rho_\Gamma V. 
$$
Besides, it can be checked that
\begin{eqnarray*}
{\cal R}_N & = & \left\{\rho \; | \; \exists \Psi \in {\cal Q}_N, \,
  \|\Psi\|_{{\cal H}_N}=1, \; \rho_\Psi=\rho \right\} \\
& = & \left\{\rho \; | \; \exists \Gamma \in {\cal S}({\cal H}_N), \,
\mbox{Ran}(\Gamma) \subset {\cal Q}_N, \;
 0 \le \Gamma \le 1,  \; \tr(\Gamma)=1, \; \rho_\Gamma=\rho
\right\} \\
& = & \left\{ \rho \ge 0 \; | \; \sqrt{\rho} \in H^1(\RR^3), \;
  \int_{\RR^3} \rho = N \right\}. 
\end{eqnarray*}
It therefore follows that
\begin{eqnarray} \label{eq:DFT_LL}
I_N & = & \inf \left\{ F_{\rm LL}(\rho) + \int_{\RR^3} \rho V, \; \rho \in {\cal
    R}_N \right\} \\
    & = & \inf \left\{ F_{\rm L}(\rho) + \int_{\RR^3} \rho V, \; \rho \in {\cal
    R}_N \right\} \label{eq:DFT_L}
\end{eqnarray}
where Levy-Lieb's and Lieb's density functionals \cite{Levy,Lieb} are
respectively defined by
\begin{eqnarray}
F_{\rm LL}(\rho) & = & \inf \left\{ \langle \Psi|H_N^1|\Psi\rangle, \; 
\Psi   \in {\cal Q}_N, \; \|\Psi\|_{{\cal H}_N}=1, \; \rho_\Psi = \rho \right\}
\label{eq:FLL} \\
F_{\rm L}(\rho) & = & \inf \big\{ \tr(H_N^1\Gamma), \; 
  \Gamma \in {\cal S}({\cal H}_N),
  \; \mbox{Ran}(\Gamma) \subset {\cal Q}_N, \nonumber \\
& & \qquad\qquad\qquad\quad
 0 \le \Gamma \le 1,  \; \tr(\Gamma)=1, \;  \rho_\Gamma = \rho
 \big\}. \label{eq:FL} 
\end{eqnarray}
Note that the functionals $F_{\rm LL}$ and $F_{\rm L}$ are independent of
the nuclear potential $V$, i.e. they do not depend on the molecular
system. They are therefore universal functionals of the density. 
It is also shown in \cite{Lieb} that $F_{\rm L}$ is the Legendre transform
of the function $V \mapsto I_N$. More precisely, expliciting the
dependency of $I_N$ on $V$, it holds
$$
F_{\rm L}(\rho) = \sup \left\{ I_N(V)-\int_{\RR^3} \rho V, \; V \in
  L^{\frac 32}(\RR^3) + L^\infty(\RR^3) \right\},
$$
from which it follows in particular that $F_{L}$ is convex on the convex
set ${\cal R}_N$ (and can be extended to a convex functional on
$L^1(\RR^3) \cap L^3(\RR^3)$).

Formulae (\ref{eq:DFT_LL}) and (\ref{eq:DFT_L}) show that, in
principle, it is possible
to compute the electronic ground state energy (and the corresponding
groud state density if it exists) by solving a minimization problem 
on ${\cal R}_N$. At this stage no approximation has been
made. But, as neither $F_{\rm LL}$ nor $F_{\rm L}$ can be easily
evaluated for the real 
system of interest ($N$ interacting electrons), approximations are needed
to make of the density functional theory a practical tool for computing
electronic ground states. Approximations rely on exact, or
very accurate, evaluations of the density functional for reference systems
``close'' to the real system:
\begin{itemize}
\item in Thomas-Fermi and related models, the reference system is an
  homogeneous electron gas;
\item in Kohn-Sham models (by far the most commonly used), it
  is a system of $N$ {\em non-interacting} electrons. 
\end{itemize}

\section{Kohn-Sham  models}
\label{sec:KS}

For a system of $N$ non-interacting electrons, universal density
functionals are obtained as 
explained in the previous section; it suffices to replace the  
interacting hamiltonian $H_N^1$ of the physical system
(formula~(\ref{eq:interacting})) with the hamiltonian of the reference system
\begin{equation} \label{eq:H0}
H_N^0 = - \sum_{i=1}^N \frac 1 2 \Delta_{\br_i}.
\end{equation}
The analogue of the Levy-Lieb density functional (\ref{eq:FLL}) then is the 
Kohn-Sham type kinetic energy functional 
\begin{equation}
  \label{eq:KS}
  \widetilde T_{\rm KS}(\rho) =\inf \left\{ \langle \Psi |H_N^0 | \Psi
    \rangle, \; 
   \Psi \in {\cal Q}_N, \; \|\Psi\|_{{\cal H}_N} = 1, \; \rho_\Psi=\rho
\right\}, 
\end{equation}
while the analogue of the Lieb functional (\ref{eq:FL}) is the Janak
kinetic energy functional  
$$
T_{\rm J}(\rho) =\inf \left\{ \tr(H_N^0 \Gamma), \; \Gamma \in {\cal
    S}({\cal H}_N), 
  \; \mbox{Ran}(\Gamma) \subset {\cal Q}_N, \;
 0 \le \Gamma \le 1,  \; \tr(\Gamma)=1, \; \rho_{\Gamma} = \rho
\right\}.  
$$
Let $\Gamma$ be in the above minimization set. The energy
$\tr(H_N^0\Gamma)$ can be rewritten as a function of the one-electron
reduced density operator $\Upsilon_\Gamma$ associated with
$\Gamma$. Recall that $\Upsilon_\Gamma$ is the self-adjoint operator on
$L^2(\RR^3_\Sigma)$ with kernel 
$$
\Upsilon_\Gamma(\bx,\bx') = N \; \int_{(\RR^3_\Sigma)^{N-1}}
\Gamma(\bx,\bx_2, \cdots, \bx_N;\bx',\bx_2, \cdots, \bx_N) \, d\bx_2
\cdots d\bx_N.
$$
Indeed, a simple calculation yields $\tr(H_N^0\Gamma) = \tr(-\frac 1 2
\Delta_\br \Upsilon_\Gamma)$, where $\Delta_\br$ is the Laplace operator on
$L^2(\RR^3_\Sigma)$ - acting on the space coordinate $\br$. Besides, it
is known (see e.g. \cite{Davidson}) that
\begin{eqnarray}
& & \left\{\Upsilon \; | \; \exists \Gamma \in {\cal
    S}({\cal H}_N), 
  \; \mbox{Ran}(\Gamma) \subset {\cal Q}_N, \;
 0 \le \Gamma \le 1,  \; \tr(\Gamma)=1, \; 
\Upsilon_{\Gamma} = \Upsilon, \; \rho_\Gamma=\rho 
\right\}  \nonumber \\
&  & = \left\{ \Upsilon \in  {\cal S}(L^2(\RR^3_\Sigma)), \; 0 \le
  \Upsilon \le 1, \; \mbox{Ran}(\Upsilon) \subset H^1(\RR^3_\Sigma), \, 
\tr(\Upsilon) = N, \; \rho_\Upsilon = \rho \right\}, \label{eq:Nrep} 
\end{eqnarray}
where 
$$
\rho_\Upsilon(\br) := \sum_{\sigma \in \Sigma} \Upsilon(\br,\sigma;\br,\sigma).
$$
Hence,
\begin{eqnarray}
T_{\rm J}(\rho) & = & \inf \bigg\{ \tr\left( -\frac 1 2 \Delta_\br \Upsilon
  \right),  \; \Upsilon \in  {\cal S}(L^2(\RR^3_\Sigma)), \; 0 \le
  \Upsilon \le 1, \nonumber \\ & & 
\qquad \qquad \qquad  \mbox{Ran}(\Upsilon) \subset H^1(\RR^3_\Sigma), \, 
\tr(\Upsilon) = N, \; \rho_\Upsilon = \rho
\bigg\}.  \label{eq:Janak1}
\end{eqnarray}
It is to be noticed that no such simple expression for $\widetilde
T_{\rm KS}(\rho)$ is available because one lacks an $N$-representation result
similar to (\ref{eq:Nrep}) for pure state one-particle reduced density
operators. In the standard Kohn-Sham model, $\widetilde T_{\rm KS}(\rho)$ is
replaced with the Kohn-Sham kinetic energy functional
\begin{equation}
  \label{eq:KS2}
  T_{\rm KS}(\rho) =\inf \left\{ \langle \Psi |H_N^0 | \Psi
    \rangle, \;
   \Psi \in {\cal Q}_N, \; \Psi \mbox{ is a Slater determinant},
   \; \rho_\Psi=\rho \right\}, 
\end{equation}
where we recall that a Slater determinant is a wavefunction $\Psi$ of
the form
$$
\Psi(\bx_1,\cdots,\bx_N) = \frac{1}{\sqrt{N!}} \mbox{det}(\phi_i(\bx_j))
\quad \mbox{with} \quad \phi_i \in L^2(\RR^3_\Sigma), \quad \int_{\RR^3}
\phi_i(\bx)\phi_j(\bx) \, d\bx = \delta_{ij}.
$$
It is then easy to check that
\begin{equation}
  \label{eq:KS2n}
   T_{\rm KS}(\rho) = \inf \left\{ {1\over 2} \sum_{i=1}^N\int_{\Rspin}
   |\nabla\phi_i(\bx)|^2 \, d\bx, \quad  \Phi=(\phi_1,\cdots,\phi_N)
   \in {\mathcal W}_N, \quad \rho_\Phi=\rho \right\},
\end{equation}
where we have set 
$$
{\mathcal W}_N = \left\{ \Phi=(\phi_1,\cdots,\phi_N) \quad | \quad
  \phi_i \in H^1(\RR^3_\Sigma), 
\; \int_{\Rspin} \phi_i(\bx) \phi_j(\bx) \, d\bx = \delta_{ij} \right\}
$$
and 
$$
\rho_\Phi(\br) = \sum_{i=1}^N \sum_{\sigma \in \Sigma} |\phi_i(\br,\sigma)|^2.
$$
Note that for an arbitrary $\rho \in  {\mathcal R}_N$, it holds
$$
T_{\rm J}(\rho)  \le \widetilde T_{\rm KS}(\rho) \le T_{\rm KS}(\rho).
$$
It is not difficult to check that (\ref{eq:Janak1}) always has a
minimizer. If one of the minimizers $\Upsilon$ of (\ref{eq:Janak1}) is of
rank $N$, then $\Upsilon = \sum_{i=1}^N |\phi_i\rangle \langle\phi_i|$
with $\Phi=(\phi_1,\cdots, \phi_N) \in {\cal W}_N$, $\Phi$ being then a
minimizer of (\ref{eq:KS2}) and  $T_{\rm KS}(\rho)=T_{\rm J}(\rho)$.  
Otherwise,  $T_{\rm KS}(\rho) >T_{\rm J}(\rho)$.

The density functionals $T_{\rm KS}$ and $T_J$ associated with the non
interacting hamiltonian $H_0$ are expected to provide acceptable
approximations of the kinetic energy of the real (interacting) system.
Likewise, the Coulomb energy
\[
J(\rho) = \frac 1 2 \int_{\RR^3} \int_{\RR^3} \frac{\rho(\br) \,
  \rho(\br')}{|\br-\br'|} \, d\br \, d\br' 
\]
representing the electrostatic energy of a {\em classical} charge
distribution of density $\rho$ is a reasonable guess for the
electronic interaction energy in a system of $N$ electrons of
density $\rho$.  
The errors on both the kinetic energy and the electrostatic interaction
are put together in the
\emph{exchange-correlation energy} defined as the difference 
\begin{equation} \label{eq:exactExc1}
E_{\rm xc}(\rho)= F_{\rm LL}(\rho)- T_{\rm KS}(\rho) - J(\rho),
\end{equation}
or
\begin{equation} \label{eq:exactExc2}
E_{\rm xc}(\rho)= F_{\rm L}(\rho)- T_{\rm J}(\rho) - J(\rho),
\end{equation}
depending on the choices for the interacting and non-interacting density
functionals. We finally end up with the so-called Kohn-Sham and extended
Kohn-Sham models
\begin{eqnarray}
  \label{eq:minKS}
  I_N^{\rm KS} & = & \inf \bigg\{ {1\over
   2}\sum_{i=1}^N\int_{\Rspin}|\nabla\phi_i(\bx)|^2 \, d\bx
 +\int_{\RR^3} \rho_\Phi V 
  + J(\rho_\Phi) +E_{\rm xc}(\rho_\Phi),  \nonumber\\   
&& \qquad \qquad  
   \Phi = (\phi_1,\cdots,\phi_N) \in {\mathcal W}_N
   \bigg\},
\end{eqnarray}
and
\begin{eqnarray}
  \label{eq:minEKS}
  I_N^{\rm EKS} & = & \inf \bigg\{  \tr\left( -\frac 1 2 \Delta_\br \Upsilon
  \right) +\int_{\RR^3} \rho_\Upsilon V 
  + J(\rho_\Upsilon) +E_{\rm xc}(\rho_\Upsilon) ,  \nonumber \\
& & \qquad  \; \Upsilon \in
  {\cal S}(L^2(\RR^3_\Sigma)), \; 0 \le \Upsilon \le 1, \; \tr(\Upsilon)
  = N, \; \tr(-\Delta_\br \Upsilon) < \infty
   \bigg\},
\end{eqnarray}
the condition on $\tr(-\Delta_\br \Upsilon)$ ensuring that each term of the
energy functional is well-defined.

\medskip

Up to now, no approximation has been made, in such a way that for the
exact exchange-correlation functionals ((\ref{eq:exactExc1}) or
(\ref{eq:exactExc2})), $I_N^{\rm KS}=I_N^{\rm EKS}=\lambda_1(H_N)$ for all
molecular system containing $N$ electrons. Unfortunately, there is no
tractable expression of $E_{\rm xc}(\rho)$ that can be used in numerical
simulations. Before proceeding further, and for the sake of simplicity,
we will restrict ourselves to closed-shell, spin-unpolarized,
systems. This means that we 
will only consider molecular systems with an even number of electrons
$N=2N_p$, where $N_p$ is the number of electron pairs in the system, and
that we will assume that electrons ``go by pairs''. 
In the Kohn-Sham formalism, this means that the set of admissible
states reduces to
$$
\left\{
  \Phi=(\varphi_1\alpha,\varphi_1\beta,\cdots,\varphi_{N_p}\alpha,\varphi_{N_p}\beta) 
  \; | \; \varphi_i \in H^1(\RR^3), \; \int_{\RR^3} \varphi_i\varphi_j =
  \delta_{ij} \right\}
$$
where $\alpha(\spinup)=1$, $\alpha(\spindown)=0$, 
$\beta(\spinup)=0$ and $\beta(\spindown)=1$, yielding the
spin-unpolarized (or closed-shell, or restricted) Kohn-Sham model
\begin{eqnarray}
  \label{eq:minRKS}
  I_N^{\rm RKS} & = & \inf \bigg\{ \sum_{i=1}^{N_p}\int_{\RR^3}|\nabla\varphi_i|^2
 +\int_{\RR^3} \rho_\varPhi V 
  + J(\rho_\varPhi) +E_{\rm xc}(\rho_\varPhi),  \nonumber\\   
&& \!\!\!\! \!\!\!\! \!\!\!\!
   \varPhi = (\varphi_1,\cdots,\varphi_{N_p}) \in (H^1(\RR^3))^{N_p}, \quad
   \int_{\RR^3} \varphi_i\varphi_j=\delta_{ij}, \quad \rho_\varPhi = 2 
  \sum_{i=1}^{N_p} |\varphi_i|^2
   \bigg\},
\end{eqnarray}
where the factor $2$ in the definition of $\rho_\varPhi$ accounts for
the spin. Likewise, the constraints on the one-electron reduced density
operators originating from the closed-shell approximation read:
$$
\Upsilon(\br,\spinup,\br',\spinup) =
\Upsilon(\br,\spindown,\br',\spindown) \quad \mbox{and} \quad 
\Upsilon(\br,\spinup,\br',\spindown) =
\Upsilon(\br,\spindown,\br',\spinup) = 0. 
$$ 
Introducing $\gamma(\br,\br') = \Upsilon(\br,\spinup,\br',\spinup)$ and
denoting by $\rho_\gamma(\br) = 2 \gamma(\br,\br)$, we
obtain the spin-unpolarized extended Kohn-Sham model
$$
I_N^{\rm REKS} = \left\{ {\cal E}(\gamma), \quad \gamma \in {\cal
    K}_{N_p} \right\} 
$$
where 
$$
{\cal E}(\gamma) = \tr(-\Delta \gamma) + 
\int_{\RR^3} \rho_\gamma V + J(\rho_\gamma) + E_{\rm xc}(\rho_\gamma),
$$
and
$$
{\cal K}_{N_p} = \left\{ \gamma \in {\cal S}(L^2(\RR^3)) \; | \; 0 \le
  \gamma \le 1, \; \tr(\gamma) = {N_p}, \; \tr(-\Delta\gamma) <
  \infty \right\}.
$$
Note that any $\gamma \in {\cal K}_{N_p}$ is of the form
$$
\gamma = \sum_{i=1}^{+\infty} n_i |\phi_i\rangle \langle \phi_i|
$$
with
$$
\phi_i \in H^1(\RR^3), \quad \int_{\RR^3} \phi_i\phi_j = \delta_{ij},
\quad 
n_i \in [0,1], \quad \sum_{i=1}^{+\infty} n_i = {N_p}, \quad
\sum_{i=1}^{+\infty} n_i \|\nabla \phi_i\|_{L^2}^2 < \infty.
$$
In particular, 
$$
\rho_\gamma(\br) = 2 \, \sum_{i=1}^{+\infty} n_i |\phi_i(\br)|^2.
$$
Let us also remark that problem (\ref{eq:minRKS}) can be recast in terms
of density operators as follows
\begin{equation} \label{eq:minRKS2}
I_N^{\rm RKS} = \left\{ {\cal E}(\gamma), \quad \gamma \in {\cal K}_{N_p} \right\}
\end{equation}
where 
$$
{\cal P}_{N_p} = \left\{ \gamma \in {\cal S}(L^2(\RR^3)) \; | \; 
  \gamma^2 = \gamma, \; \tr(\gamma) = {N_p}, \; \tr(-\Delta \gamma) <
  \infty \right\}
$$
is a the set of finite energy rank-${N_p}$ orthogonal projectors (note that
${\cal K}_{N_p}$ is the convex hull of ${\cal P}_{N_p}$). The connection between
(\ref{eq:minRKS}) and (\ref{eq:minRKS2}) is given by the correspondence
$$
\gamma = \sum_{i=1}^{N_p} |\phi_i\rangle\langle\phi_i|,
$$
i.e. $\gamma$ is the orthogonal projector on the vector space spanned by
the $\phi_i$. Indeed, as $|\nabla|=(-\Delta)^{\frac 12}$, it holds 
$$
\tr(-\Delta \gamma) = \tr(|\nabla|\gamma|\nabla|) = \sum_{i=1}^{N_p} \|
|\nabla| \phi_i \|_{L^2}^2 
= \sum_{i=1}^{N_p} \| \nabla \phi_i \|_{L^2}^2 = \sum_{i=1}^{N_p} \int_{\RR^3}
|\nabla \phi_i|^2.
$$

\medskip

Let us now address the issue of constructing relevant
approximations for $E_{\rm xc}(\rho)$.
In their celebrated 1964 article, Kohn and Sham proposed
to use an approximate exchange-correlation functional of the form
\begin{equation} \label{eq:LDA}
E_{\rm xc}(\rho) = \int_{\RR^3} g(\rho(\br)) \, d\br \qquad \mbox{(LDA
  exchange-correlation functional)} 
\end{equation}
where $\rho^{-1}g(\rho)$ is the exchange-correlation density for a
uniform electron gas with density $\rho$, yielding the so-called local
density approximation (LDA). In practical calculations, it
is made use of approximations of the function $\rho \mapsto g(\rho)$
(from $\RR_+$ to $\RR$) obtained by interpolating asymptotic
formulae for the low and high density regimes (see e.g. \cite{DG}) and accurate
quantum Monte Carlo evaluations of $g(\rho)$ for a small number of
values of $\rho$ \cite{Ceperley}. Several interpolation formulae are
available~\cite{PZ,PW,VWN}, which provide similar results.  
In the 80's, refined approximations of $E_{\rm xc}$ have been
constructed, which take into account the inhomogeneity of the electronic
density in real molecular systems. Generalized gradient approximations
(GGA) of the exchange-correlation functional are of the form
\begin{equation} \label{eq:GGA}
E_{\rm xc}(\rho) = \int_{\RR^3} h(\rho(\br),\frac 1 2|\nabla
\sqrt{\rho(\br)}|^2) \, dx 
\qquad \mbox{(GGA exchange-correlation functional)}. 
\end{equation}
Contrarily to the situation encountered for LDA, the function
$(\rho,\kappa) \mapsto g(\rho,\kappa)$ (from $\RR_+ \times \RR_+$ to
$\RR$) does not have a univoque definition. Several GGA functionals have
been proposed and new ones come up
periodically. 

\medskip

\begin{remark} \label{rem:PBE}
We have chosen the form (\ref{eq:GGA}) for the GGA
  exchange-correlation functional because it is well suited for the
  study of spin-unpolarized two electron systems (see
  Theorem~\ref{th:GGA_2e} below). 
In the Physics literature, spin-unpolarized LDA and GGA
  exchange-correlation functionals are rather written as follows
$$
E_{\rm xc}(\rho) =  E_{\rm x}(\rho) + E_{\rm c}(\rho)
$$
with
\begin{eqnarray}
E_{\rm x}(\rho) & = &  \int_{\RR^3} \rho(\br) \, \epsilon_{\rm
  x}(\rho(\br)) \, 
F_{\rm x}(s_{\rho}(\br)) \, d\br \label{eq:GGA_exchange} \\
E_{\rm c}(\rho) & = & \int_{\RR^3} \rho(\br) \, \left[ \epsilon_{\rm
  c}(r_\rho(\br)) +
H(r_\rho(\br),t_{\rho}(\br))\right] \, d\br \label{eq:GGA_correlation}. 
\end{eqnarray}
In the above decomposition, $E_{\rm x}$ is the
exchange energy, $E_{\rm c}$ is the correlation energy,
$\epsilon_{\rm x}$ and $\epsilon_{\rm c}$ are respectively the
exchange and correlation energy densities of the homogeneous electron gas,
$r_{\rho}(\br) = \left( \frac 4 3 \pi\rho(\br) \right)^{-\frac 13}$ is the
Wigner-Seitz radius, 
{ $s_\rho(\br)=\frac{1}{2(3\pi^2)^{\frac 13}} \frac{|\nabla
  \rho(\br)|}{\rho(\br)^{\frac 43}}$} is the (non-dimensional) reduced density
gradient,
{ $t_\rho(\br)=\frac{1}{4 (3\pi^{-1})^{\frac 16}} \frac{|\nabla
  \rho(\br)|}{\rho(\br)^{\frac 76}}$} is the correlation gradient, 
$F_{\rm x}$ is the so-called exchange 
enhancement factor, and $H$ is the gradient contribution to the
correlation energy. While $\epsilon_{\rm x}$ has a simple
analytical expression, namely
$$
\epsilon_{\rm x}(\rho) = -\frac 34 \left( \frac 3\pi \right)^{\frac 13} \rho^{\frac 13}
$$
$\epsilon_{\rm c}$ has to be 
approximated (as explained above for the function $g$). 
For LDA, $F_{\rm x}$ is everywhere equal to one and $H=0$. A
popular GGA exchange-correlation energy is the PBE functional
\cite{PBE}, for which  
\begin{eqnarray*}
F_{\rm x}(s) & = & 1 + \frac{\mu s^2}{1+\mu \nu ^{-1} s^2} \\
H(r,t) & = & \theta \ln \left( 1+ \frac \upsilon  \theta \, t^2 \, \frac{1+A(r)
    t^2}{1+A(r)t^2+A(r)^2t^4} \right) \quad \mbox{with} \quad
A(r) = \frac \upsilon  \theta \, \left( e^{-\epsilon_{\rm
      c}(r)/\theta}-1 \right)^{-1}, 
\end{eqnarray*}
the values of the parameters $\mu \simeq 0.21951$, $\nu \simeq 0.804$,
$\theta = \pi^{-2}(1-\ln 2)$ and $\upsilon = 3\pi^{-2}\mu$
following from theoretical arguments.
\end{remark}

\section{Main results}
\label{sec:main}

Let us first set up and comment on the conditions on the LDA and GGA
exchange-correlation functionals under which our results hold true:
\begin{itemize}
\item the function $g$ in (\ref{eq:LDA}) is a $C^1$ function from
  $\RR_+$ to $\RR$, twice differentiable and such that
\begin{eqnarray}
& & g(0) = 0 \label{eq:g0} \\
& & g' \le 0 \label{eq:gprime} \\
& & \exists 0 < \beta_- \le \beta_+ < \frac 2 3 \quad \mbox{s.t.} \quad
\sup_{\rho \in \RR_+} \frac{|g'(\rho)|}{\rho^{\beta_-}+\rho^{\beta_+}} <
\infty \label{eq:gp} \\
& & \exists 1 \le \alpha < \frac 3 2 \quad \mbox{s.t.} \quad \limsup_{\rho
  \to 0^+} \frac{g(\rho)}{\rho^\alpha} < 0; \label{eq:g0p} 
\end{eqnarray}
\item the function $h$ in (\ref{eq:LDA}) is a $C^1$ function from
  $\RR_+ \times \RR_+$ to $\RR$, twice differentiable with respect to
  the second variable, and such that
\begin{eqnarray}
& & h(0,\kappa) = 0, \; \forall \kappa \in \RR_+ \label{eq:h0} \\
& &  \frac{\partial h}{\partial \rho} \le 0 \label{eq:hprime}\\
& & \exists 0 < \beta_- \le \beta_+ < \frac 2 3 \quad \mbox{s.t.} \quad 
\sup_{(\rho,\kappa) \in \RR_+ \times \RR_+}
\frac{\dps \left|\frac{\partial h}{\partial
      \rho}(\rho,\kappa)\right|}{\rho^{\beta_-}+\rho^{\beta_+}} < \infty
\label{eq:h0infty} \\
& & \exists 1 \le \alpha < \frac 3 2 \quad \mbox{s.t.} \quad
\limsup_{(\rho,\kappa) \to (0^+,0^+)} \frac{h(\rho,\kappa)}{\rho^\alpha}
< 0 \label{eq:h0p}  \\
& & \exists 0 < a \le b < \infty \quad \mbox{s.t.} \quad
\forall (\rho,\kappa) \in \RR_+ \times \RR_+, \quad 
a \le 1 + \frac{\partial h}{\partial \kappa}(\rho,\kappa) \le b
\label{eq:h0EL} \\
& & 
\forall (\rho,\kappa) \in \RR_+ \times \RR_+, \quad 
1 + \frac{\partial h}{\partial \kappa}(\rho,\kappa) + 2 \kappa 
\frac{\partial^2 h}{\partial \kappa^2}(\rho,\kappa) \ge 0.
\label{eq:h0SCI} 
\end{eqnarray}
\end{itemize}
Conditions (\ref{eq:g0})-(\ref{eq:g0p}) on the LDA
exchange-correlation energy are not restrictive. They are obviously
fulfilled by the LDA exchange functional ($g_{\rm x}^{\rm LDA}(\rho) =  -\frac
34 \left( \frac 3\pi \right)^{\frac 13}\rho^{\frac 43}$), and are also satisfied
by all the approximate LDA correlation functionals currently used in
practice (with $\alpha = \frac 43$ and $\beta_-=\beta^+=\frac 13$).
We have checked numerically that assumptions
(\ref{eq:h0})-(\ref{eq:h0SCI}) are satisfied by the PZ81 functional defined in \cite{PZ}.

\medskip

\begin{remark}
Our results remain true if (\ref{eq:gprime}) and (\ref{eq:hprime}) are
respectively replaced with the weaker conditions
$$
\exists \frac 13 \le \beta'_- \le \beta_+ < \frac 2 3 \quad \mbox{s.t.} \quad
\sup_{\rho \in \RR_+} \frac{\max(0,g'(\rho))}{\rho^{\beta'_-}+\rho^{\beta_+}} <
\infty
$$  
and
$$
 \exists \frac 13 \le \beta'_- \le \beta_+ < \frac 2 3 \quad \mbox{s.t.} \quad 
\sup_{(\rho,\kappa) \in \RR_+ \times \RR_+}
\frac{\dps \max\left(0,\frac{\partial h}{\partial
      \rho}(\rho,\kappa)\right)}{\rho^{\beta'_-}+\rho^{\beta_+}} < \infty.
$$
\end{remark}

\medskip

As usual in the mathematical study of molecular electronic structure
models, we embed (\ref{eq:minRKS2}) in the family of problems
\begin{equation} \label{eq:EKSpb2}
I_\lambda = \inf \left\{{\cal E}(\gamma), \; \gamma \in 
{\cal K}_\lambda \right\} 
\end{equation}
parametrized by $\lambda \in \RR_+$ where
$$
{\cal K}_\lambda =  \left\{ \gamma \in {\cal S}(L^2(\RR^3))
  \; | \; 0 \le 
  \gamma \le 1, \; \tr(\gamma) = \lambda, \; \tr(- \Delta \gamma) <
  \infty \right\},
$$
and introduce the problem at infinity
\begin{equation} \label{eq:EKSpb2infty}
I^{\infty}_\lambda = \inf \left\{{\cal E}^\infty(\gamma), \; \gamma \in 
{\cal K}_\lambda \right\} 
\end{equation}
where
$$
{\cal E}^{\rm KS}(\gamma) = \tr(-\Delta \gamma) + 
J(\rho_\gamma) + E_{\rm xc}(\rho_\gamma).
$$
The following results hold true for both the LDA and GGA extended
Kohn-Sham models.

\medskip

\begin{lemma} \label{lem:I_lambda} Consider (\ref{eq:EKSpb2}) and
  (\ref{eq:EKSpb2infty}) with $E_{\rm xc}$ given either by (\ref{eq:LDA}) or
  by (\ref{eq:GGA}) together with the conditions
  (\ref{eq:g0})-(\ref{eq:g0p}) or
  (\ref{eq:h0})-(\ref{eq:h0p}). Then
  \begin{enumerate}
  \item $I_0 = I^\infty_0 = 0$ and for all $\lambda > 0$, $- \infty < I_\lambda <
    I^\infty_\lambda < 0$;
  \item the functions $\lambda \mapsto I_\lambda$ and $\lambda \mapsto
    I^\infty_\lambda$ are continuous and decreasing;
  \item for all $0 < \mu < \lambda$,
\begin{equation} \label{eq:CC1}
I_\lambda \le I_\mu + I^\infty_{\lambda-\mu}.
\end{equation}
  \end{enumerate}
\end{lemma}

\medskip

Our main results are the following two theorems.

\medskip

\begin{theorem}[Extended KS-LDA model] \label{th:LDA}
  Assume that $Z \ge N=2N_p$ (neutral
  or positively charged system) and that the function $g$ satisfies
  (\ref{eq:g0})-(\ref{eq:g0p}). Then the extended Kohn-Sham LDA model
  (\ref{eq:EKSpb2}) with $E_{\rm xc}$ given by (\ref{eq:LDA}) has a
  minimizer $\gamma_0$. Besides, $\gamma_0$ satisfies the
  self-consistent field equation 
\begin{equation} \label{eq:SCF_EKS}
\gamma_0 = \chi_{(-\infty,\epsilon_{\rm F})}(H_{\rho_{\gamma_0}}) + \delta
\end{equation} 
for some $\epsilon_{\rm F} \le 0$, where 
$$
H_{\rho_{\gamma_0}} = - \frac 1 2 \Delta + V + \rho_{\gamma_0} \star
|\br|^{-1} + g'(\rho_{\gamma_0}), 
$$
where $\chi_{(-\infty,\epsilon_{\rm
    F})}$ is the characteristic function of the range $(-\infty,\epsilon_{\rm
    F})$ and where $\delta \in {\cal
  S}(L^2(\RR^3))$ is such that $0 \le \delta \le 1$ and
$\mbox{Ran}(\delta) = \mbox{Ker}(H_{\rho_{\gamma_0}}-\epsilon_{\rm F})$.
\end{theorem}

\medskip

\begin{theorem}[Extended KS-GGA model for two electron systems]
  \label{th:GGA_2e}  
Assume that $Z \ge N=2N_p=2$ (neutral
  or positively charged system with two electrons) and that the function
  $h$ satisfies 
  (\ref{eq:h0})-(\ref{eq:h0SCI}). Then the extended Kohn-Sham GGA model
  (\ref{eq:EKSpb2}) with $E_{\rm xc}$ given by (\ref{eq:GGA}) has a
  minimizer $\gamma_0$. Besides, $\gamma_0 = |\phi \rangle \langle
  \phi|$ where $\phi$ is a minimizer of the standard spin-unpolarized
  Kohn-Sham problem (\ref{eq:minRKS}) for $N_p=1$, 
  hence satisfying the Euler equation
\begin{equation} \label{eq:EL_GGA_2e}
- \frac 1 2 \div \left( \left( 1 + \frac{\partial h}{\partial
      \kappa}(\rho_\phi,|\nabla \phi|^2)\right) \nabla \phi \right) +
\left( V + \rho_\phi \star |\br|^{-1} +  
\frac{\partial h}{\partial \rho}(\rho_\phi,|\nabla \phi|^2)
\right) \phi = \epsilon \phi  
\end{equation}
for some $\epsilon < 0$, where $\rho_\phi = 2 \phi^2$. 
In addition, $\phi \in C^{0,\alpha}(\RR^3)$ for some $0 < \alpha < 1$
and decays exponentially fast at infinity. Lastly, $\phi$ can be chosen
non-negative and $(\epsilon,\phi)$ is the lowest eigenpair of the
self-adjoint operator 
$$
- \frac 1 2 \div \left( \left( 1 + \frac{\partial h}{\partial
      \kappa}(\rho_\phi,|\nabla \phi|^2)\right) \nabla \cdot\right) +
V + \rho_\phi \star  |\br|^{-1} +  
\frac{\partial h}{\partial \rho}(\rho_\phi,|\nabla \phi|^2).
$$
\end{theorem}

\medskip

We have not been able to extend the results of Theorem~\ref{th:GGA_2e}
to the general case of $N_p$ electron pairs. This is mainly due to the
fact that the Euler equations for (\ref{eq:EKSpb2}) with $E_{\rm xc}$
given by (\ref{eq:GGA}) do not have a simple structure for $N_p \ge 2$. 

\section{Proofs}
\label{sec:proofs}

\noindent
For clarity, we will use the following notation
\begin{eqnarray*}
& & 
E_{\rm xc}^{\rm LDA}(\rho)  = \int_{\RR^3} g(\rho(\br)) \, d\br \\
\\
& &
E_{\rm xc}^{\rm GGA}(\rho)  = \int_{\RR^3} h(\rho(\br),\frac 1 2 |\nabla
\sqrt{\rho}(\br)|^2) \, d\br 
\\
& &
{\cal E}^{\rm LDA}(\gamma) = \tr(-\Delta \gamma) + 
\int_{\RR^3} \rho_\gamma V + J(\rho_\gamma) +
\int_{\RR^3} g(\rho_\gamma(\br)) \, d\br
\\
& & {\cal E}^{\rm GGA}(\gamma) = \tr(-\Delta \gamma) + 
\int_{\RR^3} \rho_\gamma V + J(\rho_\gamma) +
\int_{\RR^3} h(\rho_\gamma(\br),\frac 1 2|\nabla \sqrt{\rho_\gamma}(\br)|^2) \, d\br.
\end{eqnarray*}
The notations $E_{\rm xc}(\rho)$ and ${\cal E}(\gamma)$ will refer
indifferently to the LDA or the GGA setting.

\subsection{Preliminary results}

Most of the results of this section are elementary, but we provide them
for the sake of completeness. Let us denote by $\gS_1$ the vector space
of trace-class operators on $L^2(\RR^3)$ (see e.g. \cite{RS1}) and
introduce the vector space
$$
{\cal H} = \left\{ \gamma \in \gS_1 \; | \; |\nabla|\gamma|\nabla| \in
  \gS_1 \right\}
$$
endowed with the norm $\|\cdot\|_{\cal H} = \tr(|\cdot|) +
\tr(||\nabla|\cdot|\nabla||)$, and the convex set
$$
{\cal K} = \left\{ \gamma \in {\cal S}(L^2(\RR^3))
  \; | \; 0 \le 
  \gamma \le 1, \; \tr(\gamma) < \infty, \; \tr(|\nabla|\gamma|\nabla|) <
  \infty \right\}.
$$

\medskip

\begin{lemma} \label{lem:estim}
For all $\gamma \in {\cal K}$, $\sqrt{\rho_\gamma} \in
  H^1(\RR^3)$ and the following inequalities hold true
\begin{eqnarray}
& & \frac 1 2 \| \nabla \sqrt{\rho_\gamma} \|_{L^2}^2 \le \tr(-\Delta
\gamma) 
\label{eq:TvW} \\
& & 0 \le J(\rho_\gamma) \le C (\tr \gamma)^{\frac 32} (\tr(-\Delta
\gamma))^{\frac 12} \label{eq:UBC} \\
& & -4 Z (\tr\gamma)^{\frac 12} (\tr(-\Delta\gamma))^{\frac 12} \le \int_{\RR^3}
\rho_\gamma V \le 0 \label{eq:LBVne} \\
& & - C \left( 
  (\tr \gamma)^{1-\frac {\beta_-} 2}
  (\tr(-\Delta\gamma))^{\frac{3\beta_-}2} + 
 (\tr \gamma)^{1-\frac {\beta_+} 2}
  (\tr(-\Delta\gamma))^{\frac{3\beta_+}2} \right) \le E_{\rm xc}(\rho_\gamma)
  \le 0
\label{eq:UBExc2} \\
& & {\cal E}(\gamma)  \ge  \frac 1 2 \left( (\tr(-\Delta \gamma))^{\frac 12} - 4 Z
  (\tr \gamma)^{\frac 12} \right)^2 - 8 Z^2 \tr \gamma  \nonumber \\ 
& & \qquad \qquad \qquad - C
\left( (\tr\gamma)^{\frac{2-\beta_-}{2 - 3\beta_-}} + 
(\tr\gamma)^{\frac{2-\beta_+}{2 - 3\beta_+}} \right)  \label{eq:estimE}\\
& &  {\cal E}^\infty(\gamma)  \ge  \frac 1 2 \tr(-\Delta \gamma) - C
\left( (\tr\gamma)^{\frac{2-\beta_-}{2 - 3\beta_-}} + 
(\tr\gamma)^{\frac{2-\beta_+}{2 - 3\beta_+}} \right), \label{eq:estimEinfty}
\end{eqnarray}
for a positive constant $C$ independent of $\gamma$. 
In particular, the minimizing sequences of (\ref{eq:EKSpb2}) and those
of (\ref{eq:EKSpb2infty}) are bounded in ${\cal H}$.
\end{lemma}

\medskip

\begin{proof} Any $\gamma \in {\cal K}$ can be diagonalized in an
  orthonormal basis of $L^2(\RR^3)$ as follows
$$
\gamma = \sum_{i=1}^{+\infty} n_i |\phi\rangle \langle \phi_i|
$$
with $n_i \in [0,1]$, $\phi_i \in H^1(\RR^3)$, $\int_{\RR^3}
\phi_i\phi_j = \delta_{ij}$,  
$
\tr(\gamma) = \sum_{i=1}^{+\infty} n_i < \infty$ and $
\tr(-\Delta \gamma) = \sum_{i=1}^{+\infty} n_i \|\nabla \phi_i\|_{L^2}^2
< \infty.
$ 
As 
$$
|\nabla \sqrt{\rho_\gamma}|^2 = 2 \frac{\dps \left|\sum_{i=1}^{+\infty}
    n_i \phi_i \nabla \phi_i\right|^2}{\dps \sum_{i=1}^{+\infty} n_i
  \phi_i^2}, 
$$
(\ref{eq:TvW}) is a straightforward consequence of Cauchy-Schwarz
inequality. 
Using Hardy-Littlewood-Sobolev \cite{LiebLoss}, interpolation, and
Gagliardo-Nirenberg-Sobolev 
inequalities, we obtain
$$
J(\rho_\gamma) \le C_1 \|\rho_\gamma\|_{L^{\frac 65}}^2
\le C_1 \|\rho_\gamma\|_{L^1}^{\frac 32} \|\rho_\gamma\|_{L^3}^{\frac
  12} \le C_2 \|\rho_\gamma\|_{L^1}^{\frac 32} \|\nabla
\sqrt{\rho_\gamma}\|_{L^2}. 
$$
Hence (\ref{eq:UBC}), using (\ref{eq:TvW}) and the relation
$\|\rho_\gamma\|_{L^1} = 2 \tr(\gamma)$.  
It follows from Cauchy-Schwarz and Hardy inequalities and from the above
estimates that
$$
\int_{\RR^3} \frac{\rho_\gamma}{|\cdot -\bR_k|} \le 
2 \|\rho_\gamma\|_{L^1}^{\frac 12} \|\nabla \sqrt{\rho_\gamma}\|_{L^2} \le 4
(\tr \gamma)^{\frac 12} (\tr(-\Delta \gamma))^{\frac 12}.
$$
Hence (\ref{eq:LBVne}). Conditions~(\ref{eq:g0})-(\ref{eq:g0p}) for LDA and
(\ref{eq:h0})-(\ref{eq:h0p}) for GGA imply that $E_{\rm xc}(\rho) \le 0$
and there exists $1 < p_- <
p_+ < \frac 53$ ($p_\pm=1+\beta_\pm$) and some constant $C \in \RR_+$
such that 
\begin{equation} \label{eq:UBExc}
\forall \rho \in {\cal K}, \quad 
|E_{\rm xc}(\rho)| \le C \left( \int_{\RR^3} \rho^{p_-} + 
\int_{\RR^3} \rho^{p_+}  \right),
\end{equation}
from which we deduce (\ref{eq:UBExc2}), using interpolation and
Gagliardo-Nirenberg-Sobolev inequalities. Lastly, the estimates
(\ref{eq:estimE}) and (\ref{eq:estimEinfty}) are straightforward
consequences of (\ref{eq:UBC})-(\ref{eq:UBExc2}). 
\end{proof}

\medskip

\begin{lemma} \label{lem:density}
Let $\lambda > 0$ and $\gamma \in {\cal K}_\lambda$. There
exists a sequence $(\gamma_n)_{n \in \NN}$ such that
\begin{enumerate}
\item for all $n \in \NN$, $\gamma_n \in {\cal K}_\lambda$,
  $\gamma_n$ is finite-rank and $\mbox{Ran}(\gamma_n) \subset
  C^\infty_c(\RR^3)$;
\item $(\gamma_n)_{n \in \NN}$ converges to $\gamma$ strongly in ${\cal H}$;
\item $(\sqrt{\rho_{\gamma_n}})_{n \in \NN}$ converges to
  $\sqrt{\rho_\gamma}$ strongly in $H^1(\RR^3)$;
\item $(\rho_{\gamma_n})_{n \in \NN}$ and $(\nabla
  \sqrt{\rho_{\gamma_n}})_{n \in \NN}$ converge almost everywhere to
  $\rho_\gamma$ and  $\nabla \sqrt{\rho_\gamma}$ respectively.
\end{enumerate}
In particular
\begin{equation} \label{eq:approx}
\lim_{n \to \infty} {\cal E}(\gamma_n) = {\cal E}(\gamma) \quad
\mbox{and} \quad 
\lim_{n \to \infty} {\cal E}^\infty(\gamma_n) = {\cal E}^\infty(\gamma). 
\end{equation}
\end{lemma}

\begin{proof} Let $\gamma \in {\cal K}_\lambda$. It holds
$$
\gamma = \sum_{i=1}^{+\infty} n_i |\phi_i\rangle \langle \phi_i|
$$
with $n_i \in [0,1]$, $\phi_i \in H^1(\RR^3)$, $\int_{\RR^3}
\phi_i\phi_j = \delta_{ij}$,  
$\tr(\gamma) = \sum_{i=1}^{+\infty} n_i =\lambda$ and $
\tr(-\Delta \gamma) = \sum_{i=1}^{+\infty} n_i \|\nabla \phi_i\|_{L^2}^2
< \infty$.

We first prove that $\gamma$ can be approached by a sequence of
finite-rank operators.
Let $N_0 \in \NN$ such that $0 < n_{N_0} < 1$ (if no such
$N_0$ exists, then $\gamma$ is finite-rank and one can directly proceed
to the second part of the proof). For all $N \in
\NN$, we set
$$
\widetilde \gamma_N = \sum_{i=1}^N n_i |\phi_i\rangle \langle \phi_i| + 
\left( \lambda - \sum_{i=1}^N n_i \right) 
 |\phi_{N_0} \rangle \langle \phi_{N_0}|. 
$$
For $N$ large enough, $\widetilde \gamma_N \in {\cal
  K}_\lambda$, and the sequence $(\widetilde\gamma_N)$ obviously converges to
$\gamma$ in ${\cal H}$. Besides, $(\rho_{\widetilde\gamma_N})$ converges a.e. to
$\rho_\gamma$ and
$$
|\rho_{\widetilde\gamma_N} - \rho_\gamma | \le \left( n_{N_0} + \lambda -
  \sum_{i=1}^N n_i \right) \phi_{N_0}^2 + \sum_{i=N+1}^{+\infty} n_i
|\phi_i|^2 \le \rho_{\gamma} + \lambda  \phi_{N_0}^2.
$$ 
Hence the convergence of $(\rho_{\widetilde\gamma_N})$ to $\rho_\gamma$ in
$L^p(\RR^3)$ for all $1 \le p \le 3$. Besides, for all $N \ge N_0$,
$$
|\nabla \sqrt{\rho_{\widetilde\gamma_N}}|^2 = 2 \; \frac{\dps \left|
\sum_{i=1, \, i \neq N_0}^N n_i \phi_i \nabla\phi_i + \left(n_{N_0}+\lambda
- \sum_{i=1}^N n_i \right) \phi_{N_0} \nabla \phi_{N_0} \right|^2}
{\dps \sum_{i=1, \, i \neq N_0}^N n_i |\phi_i|^2 + \left(n_{N_0}+\lambda
- \sum_{i=1}^N n_i \right) |\phi_{N_0}|^2} 
\le 2\sum_{i=1}^{+\infty} n_i |\nabla \phi_i|^2 + 2\lambda
|\nabla\phi_{N_0}|^2. 
$$
Using Lebesgue dominated convergence theorem, we obtain that the
sequence $(\|\nabla\sqrt{\rho_{\widetilde\gamma_N}}\|_{L^2})$ converges
to $\|\nabla\sqrt{\rho_\gamma}\|_{L^2}$, from which we deduce that 
$(\sqrt{\rho_{\widetilde\gamma_N}})$ converges to
$\sqrt{\rho_\gamma}$ strongly in $H^1(\RR^3)$.

\medskip

The second part of the proof consists in approaching each $\phi_i$ by a
sequence of regular compactly supported functions. For each $i$, we
consider a sequence $(\phi_{i,k})_{k \in \NN}$ of functions of
$C^\infty_c(\RR^3)$ such that 
\begin{itemize}
\item $\mbox{supp}(\phi_{i,k}) \subset \mbox{supp}(\phi_i)$ and
  $\int_{\RR^2} \phi_{i,k} \phi_{j,k} = \delta_{ij}$ for all $k$,
\item $(\phi_{i,k})_{k \in \NN}$ converges to $\phi_i$
strongly in $H^1(\RR^3)$ and almost everywhere,
\item there exists $h_i \in L^2(\RR^3)$ such that $|\nabla \phi_{i,k}|
  \le h_i$ for all $k$.
\end{itemize}
It is then easy to check that the sequence $(\widetilde \gamma_{N,k})_{k
  \in \NN}$ defined by
$$
\widetilde \gamma_{N,k} = \sum_{i=1}^N n_i |\phi_{i,k}\rangle \langle
\phi_{i,k}| +  \left( \lambda - \sum_{i=1}^N n_i \right) 
 |\phi_{N_0,k} \rangle \langle \phi_{N_0,k}|
$$
converges to $\widetilde \gamma_N$ in ${\cal H}$ and is such that
$(\sqrt{\rho_{\widetilde \gamma_{N,k}}})_{k \in \NN}$ converges
to $\sqrt{\rho_{\widetilde \gamma_{N}}}$ strongly in $H^1(\RR^3)$.

One can then extract from $(\widetilde \gamma_{N,k})_{(N,k) \in \NN^\ast
  \times \NN}$ a subsequence $(\gamma_n)_{n \in \NN}$ which converges to
$\gamma$ in ${\cal H}$ and is such that 
$(\sqrt{\rho_{ \gamma_{n}}})_{n \in \NN}$ converges to
$\sqrt{\rho_{\gamma}}$ strongly in $H^1(\RR^3)$, and there is no
restriction in assuming that $(\rho_{\gamma_n})_{n \in \NN}$ and $(\nabla
  \sqrt{\rho_{\gamma_n}})_{n \in \NN}$ converge almost everywhere to
  $\rho_\gamma$ and  $\nabla \sqrt{\rho_\gamma}$ respectively.

The linear form $\gamma \mapsto \tr(-\Delta \gamma)$ being continuous on
${\cal H}$ and the functionals $u \mapsto \int_{\RR^3} u^2V$ and $u
\mapsto J(u^2) + E_{\rm xc}(u^2)$ being continuous on $H^1(\RR^3)$,
(\ref{eq:approx}) holds true.
\end{proof}

\subsection{Proof of Lemma~\ref{lem:I_lambda}}

Obviously, $I_0=I_0^\infty=0$ and $I_\lambda \le I^\infty_\lambda$ for
all $\lambda \in \RR_+$. 

\medskip

Let us first prove assertion~3. 
Let $0 < \mu < \lambda$, $\epsilon > 0$ and $\gamma \in {\cal
  K}_{\mu}$ such that $I_{\mu} \le {\cal E}(\gamma) \le I_\mu
+\epsilon$. It follows from Lemma~\ref{lem:density} that there is no
restriction in choosing $\gamma$ of the form
$$
\gamma = \sum_{i=1}^N n_i |\phi_i \rangle \langle \phi_i|
$$
with $0 \le n_i \le 1$, $\sum_{i=1}^N n_i = \mu$, $\langle
\phi_i|\phi_j\rangle = \delta_{ij}$ and $\phi_i \in C^\infty_c(\RR^3)$.
Likewise, there exists
$$
\gamma' = \sum_{i=1}^{N'} n_i' |\phi_i' \rangle \langle \phi_i'|
$$
with  $0 \le n_i' \le 1$, $\sum_{i=1}^{N'} n_i' = \lambda - \mu$, $\langle
\phi_i'|\phi_j'\rangle = \delta_{ij}$ and $\phi_i' \in C^\infty_c(\RR^3)$,
such that $I_{\lambda-\mu}^\infty \le {\cal E}^\infty(\gamma') \le
I_{\lambda-\mu}^\infty +\epsilon$. Let $\bf e$ be a unit vector of
$\RR^3$ and $\tau_a$ the translation operator on $L^2(\RR^3)$ defined by
$\tau_af = f(\cdot-a)$ for all $f \in L^2(\RR^3)$. For $n \in \NN$, we define
$$
\gamma_n = \gamma + \tau_{n{\bf e}} \gamma' \tau_{-n{\bf e}}.
$$
It is easy to check that for $n$ large enough, $\gamma_n \in
{\cal K}_\lambda$ and
$$
I_\lambda \le {\cal E}(\gamma_n) \le {\cal E}(\gamma) + {\cal
  E}^\infty(\gamma) +  D(\rho_{\gamma},\tau_{n{\bf e}}\rho_{\gamma'}) \le 
I_\mu + I_{\lambda-\mu}^\infty + 3\epsilon,
$$
where
{$$
D(\rho,\rho') := \int_{\RR^3} \int_{\RR^3} \frac{\rho(\br) \,
  \rho'(\br')}{|\br-\br'|} \, d\br \, d\br'. 
$$}
Hence (\ref{eq:CC1}). 

\medskip

Making use of similar arguments, it can also be proved that
\begin{equation} \label{eq:sub_infty}
I_\lambda^\infty \le I_\mu^\infty + I_{\lambda-\mu}^\infty.
\end{equation}
Let us now consider a function $\phi \in C^\infty_c(\RR^3)$ such that
$\|\phi\|_{L^2} = 1$. For all $\sigma > 0$ and all $0 \le \lambda \le 1$, the
density operator $\gamma_{\sigma,\lambda}$ with density matrix
$$
\gamma_{\sigma,\lambda}(\br,\br') = \lambda \sigma^3 \, \phi(\sigma \br) \,
\phi(\sigma \br') 
$$
is in ${\cal K}_\lambda$. Using (\ref{eq:g0p}) for LDA and
(\ref{eq:h0p}) for GGA, we obtain that there exists
$1 \le \alpha < \frac 32$,
$c > 0$ and $\sigma_0 > 0$ such that for all $0 \le \lambda \le 1$ and all
$0 \le  \sigma \le \sigma_0$,
$$
I^\infty_\lambda \le {\cal E}^\infty(\gamma_{\sigma,\lambda}) \le
\lambda \sigma^2 \int_{\RR^3} |\nabla \phi|^2 + \lambda^2 \sigma
J(2 |\phi|^2) - c \lambda^{\alpha} \sigma^{3(\alpha-1)} 
\int_{\RR^3} |\phi|^{2\alpha}. 
$$
Therefore $I^\infty_\lambda < 0$ for $\lambda$ positive and small
enough. It follows from (\ref{eq:CC1}) and (\ref{eq:sub_infty}) that the
functions $\lambda \mapsto I_\lambda$ and $\lambda \mapsto
I_\lambda^\infty$ are decreasing, and that for all $\lambda > 0$,
$$
-\infty < I_\lambda \le I_\lambda^\infty < 0.
$$
To proceed further, we need the following lemma.

\medskip

\begin{lemma} \label{lem:non_vanishing}
Let $\lambda > 0$ and $(\gamma_n)_{n \in \NN}$ be a
  minimizing sequence for (\ref{eq:EKSpb2}). Then the sequence
  $(\rho_{\gamma_n})_{n \in \NN}$ cannot vanish, which means that
$$
\exists R > 0 \quad \mbox{s.t.} \quad \lim_{n \to \infty} \sup_{x \in
  \RR^3} \int_{x+B_R} \rho_{\gamma_n} > 0.
$$ 
The same holds true for the minimizing sequences of (\ref{eq:EKSpb2infty}).
\end{lemma}

\medskip

\begin{proof} Let $(\gamma_n)_{n \in \NN}$ be a minimizing sequence for
  (\ref{eq:EKSpb2}). By contradiction, assume that
$$
\forall R > 0, \quad \lim_{n \to \infty} \sup_{x \in
  \RR^3} \int_{x+B_R} \rho_n = 0.
$$ 
Let $1 < p < \frac 53$. For $\rho \ge 0$ such that $\sqrt{\rho} \in
H^1(\RR^3)$, it holds for all $k \in \ZZ^3$,
$$
\int_{k+B_1} \rho^{p} \le \left( \int_{k+B_1} \rho \right)^{p-1} 
\left( \int_{k+B_1}  \rho^{\frac{1}{2-p}} \right)^{2-p} 
\le C_p \left( \int_{k+B_1} \rho \right)^{p-1} 
\left( \int_{k+B_1} (\rho + |\nabla \sqrt{\rho}|^2) \right) 
$$
(where the constant $C_p$ does not depend on $k$). We therefore obtain
\begin{eqnarray*}
\int_{\RR^3} \rho^{p} & \le & \sum_{k \in \ZZ^3} \int_{k+B_1}
\rho^{p} \\
& \le & C_p  \sum_{k \in \ZZ^3} \left( \int_{k+B_1} \rho \right)^{p-1} 
\left( \int_{k+B_1} (\rho + |\nabla \sqrt{\rho}|^2) \right)  \\
& \le & 8 C_p \left(\sup_{x \in \RR^3} \int_{x+B_1} \rho\right)^{p-1}  
\left( \int_{\RR^3} (\rho + \int_{\RR^3} |\nabla \sqrt \rho|^2) \right).
\end{eqnarray*}
Hence, for all $\gamma \in {\cal K}$,
$$
\int_{\RR^3} \rho_{\gamma}^{p} \le 
16 C_p \left(\sup_{x \in \RR^3} \int_{x+B_1} \rho_\gamma\right)^{p-1}
\, \| \gamma \|_{\cal H}^2.  
$$
As we know that any minimizing sequence of (\ref{eq:EKSpb2}) is bounded in
${\cal H}$, we deduce from the above inequality that for all $1 < p <
\frac 53$,  
$$
\lim_{n \to \infty} \int_{\RR^3} \rho_{\gamma_n}^{p} = 0.
$$
In particular, it follows from (\ref{eq:UBExc}) that 
$$
\lim_{n \to \infty} \int_{\RR^3} E_{\rm xc}(\rho_{\gamma_n}) = 0.
$$
Let us now fix $1 < p < \frac 32$, $\epsilon > 0$ and $R > 0$ such that
$|V| \le \epsilon\lambda^{-1}$ on $B_R^c$. For $n$ large enough, we have
$$
\left| \int_{\RR^3} \rho_{\gamma_n} V \right| \le 
 \int_{B_R} \rho_{\gamma_n} |V| + \int_{B_R^c} \rho_{\gamma_n} |V| 
\le \left( \int_{B_R} |V|^{p'} \right)^{\frac{1}{p'}} 
\left( \int_{B_R} \rho_{\gamma_n}^p \right)^{\frac 1p} 
+ \frac{\epsilon}{\lambda} \int_{B_R^c} \rho_{\gamma_n} \le 2 \epsilon.
$$
Therefore 
$$
\lim_{n \to \infty} \int_{\RR^3} \rho_{\gamma_n} V = 0.
$$
As,
$$
{\cal E}(\gamma_n) \ge \int_{\RR^3} \rho_{\gamma_n} V + E_{\rm
  xc}(\rho_{\gamma_n}),
$$
we obtain that $I_\lambda \ge 0$. This is in contradiction with the
previously proved result stating that $I_\lambda < 0$. Hence
$(\rho_{\gamma_n})_{n \in \NN}$ cannot vanish. The case of problem
(\ref{eq:EKSpb2infty}) is easier since the only non-positive term in the
energy functional is $E_{\rm xc}(\rho)$. 
\end{proof}

\noindent 
We can now prove that $I_\lambda < I_\lambda^\infty$. For this purpose 
let us consider a minimizing sequence $(\gamma_n)_{n \in \NN}$ for 
(\ref{eq:EKSpb2infty}).  
We deduce from Lemma~\ref{lem:non_vanishing} that there exists $\eta
>0$ and $R > 0$, such that for $n$ large enough, there exists $x_n \in
\RR^3$ such that
$$
\int_{x_n + B_R} \rho_{\gamma_n} \ge \eta.
$$
Let us introduce $\widetilde \gamma_n = \tau_{\bar x_1 - x_n} \gamma_n
\tau_{x_n-\bar x_1}$. Clearly $\widetilde \gamma_n \in
{\cal K}_\lambda$ and
$$
{\cal E}(\widetilde \gamma_n) \le {\cal E}^\infty(\gamma_n) - \frac{z_1\eta}R.
$$ 
Thus,
$$
I_\lambda \le I_\lambda^\infty - \frac{z_1\eta}R < I_\lambda^\infty.
$$
It remains to prove that the functions $\lambda \mapsto I_\lambda$ and
$\lambda \mapsto I_\lambda^\infty$ are continuous. We will deal here
with the former one, the same arguments applying to the latter one. The
proof is based on the following lemma.

\medskip

\begin{lemma} \label{lem:scalingExc}
  Let $(\alpha_k)_{k \in \NN}$ be a sequence of positive
  real numbers converging to $1$, and $(\rho_k)_{k \in \NN}$ a sequence
  of non-negative densities such that $(\sqrt{\rho_k})_{k \in \NN}$ is
  bounded in $H^1(\RR^3)$. Then
$$
\lim_{k \to \infty} \left( E_{\rm xc}(\alpha_k\rho_k)-E_{\rm xc}(\rho_k)
  \right) = 0.
$$
\end{lemma}

\medskip

\begin{proof}
In the LDA setting, we deduce from (\ref{eq:g0p}) that there exists $1
< p_- \le p_+ < \frac 53$ and $C \in \RR_+$ such that for $k$ large
enough
$$
\left| E_{\rm xc}^{\rm LDA}(\alpha_k\rho_k)-E_{\rm xc}^{\rm LDA}(\rho_k) \right|
\le C \left|\alpha_k-1\right| \int_{\RR^3} (\rho_k^{p_-}+\rho_k^{p_+}).
$$
In the GGA setting, we obtain from (\ref{eq:h0infty}) and
(\ref{eq:h0EL}) that there 
exists $1 < p_- \le p_+ < \frac 53$ and $C \in \RR_+$ such that for $k$
large enough
$$
\left| E_{\rm xc}^{\rm GGA}(\alpha_k\rho_k)-E_{\rm xc}^{\rm GGA}(\rho_k) \right|
\le C \left|\alpha_k-1\right| \int_{\RR^3}
(\rho_k^{p_-}+\rho_k^{p_+}+|\nabla \sqrt{\rho_k}|^2). 
$$
As $(\sqrt{\rho_k})_{k \in \NN}$ is bounded in $H^1(\RR^3)$,
$(\rho_k)_{k \in \NN}$ is bounded in $L^p(\RR^3)$ for all $1 \le p \le
3$ and $(\nabla\sqrt{\rho_k})_{k \in \NN}$ is bounded in
$(L^2(\RR^3))^3$, hence the result. 
\end{proof}

\medskip

We can now complete the proof of Lemma~\ref{lem:I_lambda}.

\medskip

\noindent
{\it Left-continuity of $\lambda \mapsto I_\lambda$.} Let $\lambda > 0$, and $(\lambda_k)_{k \in \NN}$
be an increasing sequence of positive real numbers converging to
$\lambda$. Let $\epsilon > 0$ and $\gamma \in {\cal
  K}_\lambda$ such that  
$$
I_\lambda \le {\cal E}(\gamma) \le I_\lambda + \frac \epsilon 2.
$$
For all $k \in \NN$, $\gamma_k = \lambda_k \lambda^{-1} \gamma$ is in
${\cal K}_{\lambda_k}$ so that
$$
\forall k \in \NN, \quad \forall n \in \NN, \quad 
I_\lambda \le I_{\lambda_k} \le {\cal E}(\gamma_{k}).   
$$
Besides,
$$
{\cal E}(\gamma_{k}) = \frac{\lambda_k}{\lambda} \tr(-\Delta\gamma) + 
\frac{\lambda_k}{\lambda} \int_{\RR^3} \rho_\gamma V + 
\frac{\lambda_k^2}{\lambda^2} J(\rho_\gamma) +
E_{\rm xc}\left( \frac{\lambda_k}{\lambda} \rho_\gamma \right)
\mathop{\longrightarrow}_{k \to \infty} {\cal E}(\gamma)
$$
in virtue of Lemma~\ref{lem:scalingExc}. Thus 
$$
I_\lambda \le I_{\lambda_k} \le I_\lambda + \epsilon
$$
for $k$ large enough.

\medskip

\noindent
{\it Right-continuity of $\lambda \mapsto I_\lambda$.} Let $\lambda >
0$, and $(\lambda_k)_{k \in \NN}$ 
be an decreasing sequence of positive real numbers converging to
$\lambda$. For each $k \in \NN$, we choose $\gamma_k \in 
{\cal K}_{\lambda_k}$ such that 
$$
I_{\lambda_k} \le {\cal E}(\gamma_k) \le I_{\lambda_k}+\frac 1 k.
$$
For all $k \in \NN$, we set $\widetilde \gamma_k = \lambda
\lambda_k^{-1} \gamma_k$. As $\widetilde \gamma_k \in 
{\cal K}_\lambda$, it holds
$$
I_\lambda \le {\cal E}(\widetilde \gamma_k) =
\frac{\lambda}{\lambda_k} \tr(-\Delta\gamma_k) + 
\frac{\lambda}{\lambda_k} \int_{\RR^3} \rho_{\gamma_k} V + 
\frac{\lambda}{\lambda_k^2} J(\rho_{\gamma_k}) +
E_{\rm xc}\left( \frac{\lambda}{\lambda_k} \rho_{\gamma_k} \right).
$$
As $(\gamma_k)_{k \in \NN}$ is bounded in ${\cal H}$ and
$(\sqrt{\rho_{\gamma_k}})_{k \in \NN}$ is bounded in $H^1(\RR^3)$, we
deduce from Lemma~\ref{lem:scalingExc} that
$$
\lim_{k \to \infty} \left( {\cal E}(\widetilde \gamma_k) - {\cal
    E}(\gamma_k) \right) = 0.
$$
Let $\epsilon > 0$ and $k_\epsilon \ge 2\epsilon^{-1}$ such that for all $k \ge
k_\epsilon$, 
$$
\left|  {\cal E}(\widetilde \gamma_k) - {\cal E}(\gamma_k) \right| \le
\frac\epsilon 2.
$$
Then,
$$
\forall k \ge k_\epsilon, \quad 
I_\lambda - \epsilon \le I_{\lambda_k} \le I_\lambda.
$$
This proves the right-continuity of $\lambda \mapsto I_\lambda$ on
$\RR_+ \setminus \left\{0\right\}$.
Lastly, it results from the estimates established in
Lemma\ref{lem:estim} that  
$$
\lim_{\lambda \to 0^+} I_\lambda = 0.
$$

\subsection{Proof of Theorem~\ref{th:LDA}}

Let us first prove the following lemma.

\medskip

\begin{lemma} \label{lem:conv_LDA}
Let $(\gamma_n)_{n \in \NN}$ be a sequence of elements of
  ${\cal K}$, bounded in ${\cal H}$, which converges to $\gamma$ for the
  weak-$*$ topology of ${\cal H}$. If $\lim_{n\to\infty}\tr(\gamma_n) =
  \tr(\gamma)$, 
  then $(\rho_{\gamma_n})_{n \in \NN}$ converges
  to $\rho_\gamma$ strongly in $L^p(\RR^3)$ for all $1 \le p < 3$ and
$$
{\cal E}^{\rm LDA}(\gamma) \le \liminf_{n \to \infty} {\cal E}^{\rm
  LDA}(\gamma_n) \quad \mbox{and} \quad 
{\cal E}^{{\rm LDA},\infty}(\gamma) \le \liminf_{n \to \infty} {\cal E}^{{\rm
  LDA},\infty}(\gamma_n).
$$
\end{lemma}

\medskip

\begin{proof}
The fact that $(\gamma_n)_{n \in \NN}$ converges to $\gamma$ for the
weak-$*$ topology of ${\cal H}$ means that for all compact operator $K$ on
$L^2(\RR^3)$, 
$$
\lim_{n \to \infty}\tr(\gamma_nK) = \tr(\gamma K) \quad \mbox{and} \quad
\lim_{n \to \infty}\tr(|\nabla|\gamma_n|\nabla|K) =
\tr(|\nabla|\gamma|\nabla| K). 
$$
For all $W \in C^\infty_c(\RR^3)$, the operator 
$(1+|\nabla|)^{-1}W(1+|\nabla|)^{-1}$ is compact (it is even in the
Schatten class $\gS_p$
for all $p > \frac 32$ in virtue of the Kato-Seiler-Simon
inequality~\cite{Simon}), yielding 
\begin{eqnarray*} 
\int_{\RR^3} \rho_{\gamma_n} W & = & 2 \, \tr(\gamma_nW) = 2 \, 
\tr((1+|\nabla|)\gamma_n(1+|\nabla|)(1+|\nabla|)^{-1}W(1+|\nabla|)^{-1})
\\
& \dps \mathop{\rightarrow}_{n \to \infty} & 2 \, 
\tr((1+|\nabla|)\gamma(1+|\nabla|)(1+|\nabla|)^{-1}W(1+|\nabla|)^{-1})= 2 \, \tr(\gamma W) = \int_{\RR^3} \rho_{\gamma}W
\end{eqnarray*}
Hence, $(\rho_{\gamma_n})_{n \in \NN}$ converges to $\rho_\gamma$ in
${\cal D}'(\RR^3)$. As by (\ref{eq:TvW}), $(\sqrt{\rho_{\gamma_n}})_{n \in
  \NN}$ is bounded in $H^1(\RR^3)$, it follows that
$(\sqrt{\rho_{\gamma_n}})_{n \in \NN}$ converges to $\sqrt{\rho_\gamma}$
weakly in $H^1(\RR^3)$, and strongly in $L^p_{\rm loc}(\RR^3)$ for all $2
\le p < 6$. In particular, $(\sqrt{\rho_{\gamma_n}})_{n \in \NN}$
converges to $\sqrt{\rho_\gamma}$ weakly in $L^2(\RR^3)$. But we also
know that 
$$
\lim_{n \to \infty} \|\sqrt{\rho_{\gamma_n}}\|_{L^2}^2 = 
\lim_{n \to \infty}  \int_{\RR^3} \rho_{\gamma_n} =  2
\lim_{n\to\infty}\tr(\gamma_n) = 2  \tr(\gamma) = \int_{\RR^3} \rho_\gamma
= \|\sqrt{\rho_{\gamma}}\|_{L^2}^2.
$$
Therefore, the convergence of $(\sqrt{\rho_{\gamma_n}})_{n \in \NN}$
to $\sqrt{\rho_\gamma}$ holds strongly in $L^2(\RR^3)$. By an elementary
bootstrap
argument exploiting the boundedness of $(\sqrt{\rho_{\gamma_n}})_{n \in
  \NN}$ in $H^1(\RR^3)$, we obtain that  $(\sqrt{\rho_{\gamma_n}})_{n
  \in \NN}$ converges strongly to $\sqrt{\rho_\gamma}$ in $L^p(\RR^3)$
for all $2 \le p < 6$, hence that $(\rho_{\gamma_n})_{n \in \NN}$
converges to $\rho_\gamma$ strongly in $L^p(\RR^3)$ for all $1 \le p <
3$. This readily implies
\begin{eqnarray*}
&& \lim_{n \to \infty} \int_{\RR^3} \rho_{\gamma_n} V = 
\int_{\RR^3} \rho_{\gamma} V \\
&& \lim_{n \to \infty} J(\rho_{\gamma_n}) =
J(\rho_\gamma) \\
&& \lim_{n \to \infty} E_{\rm xc}^{\rm LDA}(\rho_{\gamma_n}) =
E_{\rm xc}^{\rm LDA}(\rho_{\gamma}) . 
\end{eqnarray*}
Lastly, for any orthonormal basis $(\psi_k)_{k \in \NN^\ast}$ of
$L^2(\RR^3)$ such that $\psi_k \in H^1(\RR^3)$ for all $k$, we have 
\begin{eqnarray*}
\tr(|\nabla|\gamma|\nabla|) & = & \sum_{k=1}^{+\infty} 
\langle \psi_k||\nabla|\gamma|\nabla||\psi_k \rangle \\
& = &  \sum_{k=1}^{+\infty} \tr(\gamma (||\nabla|\psi_k\rangle \langle
|\nabla|\psi_k|)) \\
& = & \sum_{k=1}^{+\infty} \lim_{n \to \infty} 
\tr(\gamma_n (||\nabla|\psi_k\rangle \langle
|\nabla|\psi_k|)) \\
& \le &  \liminf_{n \to \infty}   \sum_{k=1}^{+\infty} \tr(\gamma_n
(||\nabla|\psi_k\rangle \langle |\nabla|\psi_k|)) \\
& = &  \liminf_{n \to \infty} \tr(|\nabla|\gamma_n|\nabla|). 
\end{eqnarray*}
We thus obtain the desired result.
\end{proof}

\noindent
We are now in position to prove Theorem~\ref{th:LDA}.
Let $(\gamma_n)_{n \in \NN}$ be a minimizing sequence for
$I_\lambda$. We know from Lemma~\ref{lem:estim} that
$(\gamma_n)_{n \in \NN}$ is bounded in $\cal H$ and that
$(\sqrt{\rho_{\gamma_n}})_{n \in \NN}$ is bounded 
in $H^1(\RR^3)$. Replacing $(\gamma_n)_{n \in \NN}$ by a suitable
subsequence, we can assume that $(\gamma_n)$ converges to some $\gamma
\in {\cal K}$ for the weak-$\ast$ topology of ${\cal H}$ and that
$(\sqrt{\rho_{\gamma_n}})_{n \in \NN}$ converges to $\sqrt{\rho_\gamma}$
weakly in $H^1(\RR^3)$, strongly in $L^p_{\rm loc}(\RR^3)$ for all $2
\le p < 6$ and almost everywhere.

\medskip

If $\tr(\gamma) = \lambda$, then $\gamma \in {\cal K}_\lambda$
and according to Lemma~\ref{lem:conv_LDA}, 
$$
{\cal E}^{\rm LDA}(\gamma) \le \liminf_{n \to +\infty} {\cal E}^{\rm
  LDA}(\gamma_n) = I_\lambda 
$$
yielding that $\gamma$ is a minimizer of (\ref{eq:EKSpb2}).

\medskip

The rest of the proof consists in rulling out the eventuality when
$\tr(\gamma) < 
\lambda$. Let us therefore set $\alpha = \tr(\gamma)$ and assume that 
$0 \le \alpha < \lambda$. Following e.g. \cite{FLSS}, we consider a quadratic
partition of the unity $\xi^2+\chi^2=1$, where $\xi$ is a smooth,
radial function, nonincreasing in the radial direction, such that
$\xi(0)=1$, $0 \le \xi(x) < 1$ if $|x| > 0$, $\xi(x) = 0$ if $|x|
\ge 1$, $\|\nabla \xi\|_{L^\infty} \le 2$ and $\|\nabla (1-\xi^2)^{\frac
  12} \|_{L^\infty} \le 2$. We then set $\xi_R(\cdot) = \xi \left(
  \frac{\cdot}{R} 
\right)$. For all $n \in \NN$, $R \mapsto \tr(\xi_R\gamma_n\xi_R)$ is
a continuous nondecreasing function which vanishes at $R=0$ and
converges to $\tr(\gamma_n) = \lambda$ when $R$ goes to infinity. Let
$R_n > 0$ be such that $\tr(\xi_{R_n}\gamma_n\xi_{R_n})=\alpha$.
The
sequence $(R_n)_{n \in \NN}$ goes to infinity; otherwise, it would
contain a subsequence $(R_{n_k})_{k \in \NN}$ converging to a
finite value $R^\ast$, and we would then get
$$
\int_{\RR^3} \rho_\gamma(x) \xi_{R^\ast}^2(x) \, dx =
\lim_{n \to \infty}  \int_{\RR^3} \rho_{\gamma_n}(x) \xi_{R_n}^2(x) \, dx
= 2\lim_{n\to \infty} \tr(\xi_{R_n}\gamma_n\xi_{R_n}) = 2\alpha = \int_{\RR^3}
\rho_\gamma(x) \, dx. 
$$
As $\xi_{R^\ast}^2 < 1$ on $\RR^3 \setminus \left\{0\right\}$, we reach
a contradiction. Consequently, $(R_n)_{n \in \NN}$ indeed goes to
infinity. Let us now introduce
$$
\gamma_{1,n} = \xi_{R_n}\gamma_n\xi_{R_n} \quad \mbox{and} \quad
\gamma_{2,n} = \chi_{R_n}\gamma_n\chi_{R_n}.
$$
Note that $\gamma_{1,n}$ and $\gamma_{2,n}$ are trace-class
self-adjoint operators on $L^2(\RR^3)$ such that $0 \le \gamma_{j,n} \le
1$, that $\rho_{\gamma_n} = \rho_{\gamma_{1,n}}+ \rho_{\gamma_{2,n}}$
and that $\tr(\gamma_{1,n}) = \alpha$ while $\tr(\gamma_{2,n}) = \lambda
- \alpha$. Besides, using the IMS formula
$$
-\Delta = \chi_{R_n}(-\Delta)\chi_{R_n} + \xi_{R_n}(-\Delta)\xi_{R_n} 
-|\nabla \chi_{R_n}|^2 -|\nabla \xi_{R_n}|^2,   
$$
it holds
\begin{eqnarray}
\tr(-\Delta \gamma_n) & = & \tr(- \Delta \gamma_{1,n}) +
\tr(- \Delta \gamma_{2,n}) - \tr((|\nabla \chi_{R_n}|^2 +|\nabla
\xi_{R_n}|^2) \gamma_n) \nonumber \\
& \ge &  \tr(- \Delta \gamma_{1,n}) + \tr(- \Delta \gamma_{2,n})
- \frac{4\lambda}{R_n^2}, \label{eq:IMS}
\end{eqnarray}
from which we infer that both $(\gamma_{1,n})_{n \in \NN}$ and 
 $(\gamma_{2,n})_{n \in \NN}$ are bounded sequences of ${\cal H}$. As
 for all $\phi \in C^\infty_c(\RR^3)$, 
\begin{eqnarray*}
\tr(\gamma_{1,n}(|\phi\rangle\langle\phi|)) & = & 
\tr(\gamma_n(|\xi_{R_n}\phi\rangle\langle\xi_{R_n}\phi|)) \\
& = & 
\tr(\gamma_n(|(\xi_{R_n}-1)\phi\rangle\langle\xi_{R_n}\phi|)) 
 +\tr(\gamma_n(|\phi\rangle\langle(\xi_{R_n}-1)\phi|)) 
 +\tr(\gamma_n(|\phi\rangle\langle\phi|)) \\
& \dps \mathop{\longrightarrow}_{n
  \to \infty} & \tr(\gamma (|\phi\rangle\langle\phi|)),
\end{eqnarray*}
we obtain that $(\gamma_{1,n})_{n \in \NN}$ converges to $\gamma$ for
the weak-$*$ topology of ${\cal H}$. Since $\tr(\gamma_{1,n}) = \alpha =
\tr(\gamma)$ for all $n$, we deduce from Lemma~\ref{lem:conv_LDA} that 
$(\rho_{\gamma_{1,n}})_{n \in \NN}$ converges to $\rho_\gamma$ strongly
in $L^p(\RR^3)$ for all $1 \le p < 3$, and that
\begin{equation} \label{eq:LSI_LDA}
{\cal E}^{\rm LDA}(\gamma) \le \lim_{n \to \infty} {\cal E}^{\rm LDA}(\gamma_{1,n}).
\end{equation}
As a by-product, we also obtain that $(\rho_{\gamma_{2,n}})_{n \in \NN}$
converges strongly to zero in $L^p_{\rm loc}(\RR^3)$ for all $1 \le p <
3$ (since $\rho_{\gamma_{2,n}} = \rho_{\gamma_n}-\rho_{\gamma_{1,n}}$
with $(\rho_{\gamma_{n}})_{n \in \NN}$ and $(\rho_{\gamma_{1,n}})_{n \in
  \NN}$ both converging to $\rho_{\gamma}$ in $L^p_{\rm loc}(\RR^3)$ for
all $1 \le p < 3$).
Besides, using again (\ref{eq:IMS}), it holds
\begin{eqnarray*}
{\cal E}^{\rm LDA}(\gamma_n) & = & \tr(-\Delta \gamma_n) + \int_{\RR^3}
\rho_{\gamma_n} V + J(\rho_{\gamma_n}) +
\int_{\RR^3} g(\rho_{\gamma_n}) \\
& \ge &  \tr(- \Delta \gamma_{1,n}) + \tr(- \Delta \gamma_{2,n})
+  \int_{\RR^3} \rho_{\gamma_{1,n}} V +  \int_{\RR^3}
\rho_{\gamma_{2,n}} V \\ & & + J(\rho_{\gamma_{1,n}}) +
J(\rho_{\gamma_{2,n}}) +  \int_{\RR^3} 
g(\rho_{\gamma_{1,n}}+\rho_{\gamma_{2,n}}) -  \frac{4\lambda}{R_n^2} \\
& = & {\cal E}^{\rm LDA}(\gamma_{1,n}) + {\cal E}^{{\rm LDA},\infty}(\gamma_{2,n}) +
\int_{\RR^3} \rho_{\gamma_{2,n}} V \\ 
& &  + \int_{\RR^3} (g(\rho_{\gamma_{1,n}}+\rho_{\gamma_{2,n}}) 
- g(\rho_{\gamma_{1,n}}) - g(\rho_{\gamma_{2,n}})) -  \frac{4\lambda}{R_n^2}.
\end{eqnarray*}
For $R$ large enough, one has on the one hand
$$
\left| \int_{\RR^3} \rho_{\gamma_{2,n}} V \right| 
\le 2 Z \left( \int_{B_R} \rho_{\gamma_{2,n}} \right)^{\frac 12}
\|\nabla\sqrt{\rho_{\gamma_{2,n}}}\|_{L^2}  
+ \frac{2Z(\lambda-\alpha)} R,
$$
and on the other hand
\begin{eqnarray*}
\left| \int_{\RR^3} (g(\rho_{\gamma_{1,n}}+\rho_{\gamma_{2,n}}) 
- g(\rho_{\gamma_{1,n}}) - g(\rho_{\gamma_{2,n}})) \right| & \le & 
 \int_{B_R} \left| g(\rho_{\gamma_{1,n}}+\rho_{\gamma_{2,n}}) -
   g(\rho_{\gamma_{1,n}}) \right| + 
\int_{B_R} \left| g(\rho_{\gamma_{2,n}}) \right| \\ & + & 
 \int_{B_R^c} \left| g(\rho_{\gamma_{1,n}}+\rho_{\gamma_{2,n}}) - 
g(\rho_{\gamma_{2,n}}) \right| +
\int_{B_R^c} \left| g(\rho_{\gamma_{1,n}}) \right| \\
& \le & C \left( \int_{B_R} (\rho_{\gamma_{2,n}}+\rho_{\gamma_{2,n}}^2)
  + \|\rho_{\gamma_{1,n}}\|_{L^2} 
\left(  \int_{B_R} \rho_{\gamma_{2,n}}^2 \right)^{\frac 12} \right) \\
& +  &  C \left( \int_{B_R} \rho_{\gamma_{2,n}}^{p_-} +
  \rho_{\gamma_{2,n}}^{p_+}  \right)    \\
& +  &  C \left( \int_{B_R^c}
  (\rho_{\gamma_{1,n}}+\rho_{\gamma_{1,n}}^2) +
  \|\rho_{\gamma_{2,n}}\|_{L^2} 
\left(  \int_{B_R^c} \rho_{\gamma_{1,n}}^2 \right)^{\frac 12} \right) \\
& +  &  C \left( \int_{B_R^c} \rho_{\gamma_{1,n}}^{p_-} +
  \rho_{\gamma_{1,n}}^{p_+} \right) 
\end{eqnarray*}
for some constant $C$ independent of $R$ and $n$.
Yet, we know that $(\sqrt{\rho_{\gamma_{1,n}}})_{n \in \NN}$ and
$(\sqrt{\rho_{\gamma_{1,n}}})_{n \in \NN}$ are bounded in $H^1(\RR^3)$, that 
$(\rho_{\gamma_{1,n}})_{n \in \NN}$ converges to $\rho_{\gamma}$ in
$L^p(\RR^3)$ for all $1 \le p < 3$ and that 
$(\rho_{\gamma_{2,n}})_{n \in \NN}$ converges to $0$ in
$L^p_{\rm loc}(\RR^3)$ for all $1 \le p < 3$.
Consequently, there exists for all $\epsilon > 0$, some $N \in \NN$ such
that for all $n \ge N$, 
$$
{\cal E}^{\rm LDA}(\gamma_n) \ge  {\cal E}^{\rm LDA}(\gamma_{1,n}) + {\cal
  E}^{{\rm LDA},\infty}(\gamma_{2,n}) - \epsilon \ge  I_\alpha +
I^\infty_{\lambda-\alpha} - \epsilon. 
$$
Letting $n$ go to infinity, $\epsilon$ go to zero, and using
(\ref{eq:CC1}), we obtain that $I_\lambda = 
I_\alpha + I^\infty_{\lambda-\alpha}$ and that $(\gamma_{1,n})_{n \in \NN}$ 
and $(\gamma_{2,n})_{n \in \NN}$ are minimizing sequences for $I_\alpha$
and $I^\infty_{\lambda-\alpha}$ respectively. It also follows from
(\ref{eq:LSI_LDA}) that $\gamma$ is a minimizer for $I_\alpha$.
In
particular $\gamma$ satisfies the Euler equation
$$
\gamma = 1_{(-\infty,\epsilon_{\rm F})}(H_{\rho_\gamma}) + \delta
$$
for some Fermi level $\epsilon_{\rm F} \in \RR$, where 
$$
H_{\rho_\gamma} = -\frac 1 2 \Delta + V + \rho_\gamma \star
|\br|^{-1} + g'(\rho_\gamma),
$$
and where $0 \le \delta \le 1$, $\mbox{Ran}(\delta) \subset
\mbox{Ker}(H_{\rho_\gamma}-\epsilon_{\rm F})$. 
As $V + \rho_\gamma \star
|\br|^{-1} + g'(\rho_\gamma)$ is $\Delta$-compact, the
essential spectrum of $H_{\rho_\gamma}$ is $[0,+\infty)$. Besides, 
$H_{\rho_\gamma}$ is bounded from below, 
$$
H_{\rho_\gamma} \le  -\frac 1 2 \Delta + V + \rho_\gamma \star
|\br|^{-1} ,
$$
and we know from \cite[Lemma~II.1]{Lions} that as $-\sum_{k=1}^M z_k + \int_{\RR^3}
\rho_\gamma = -Z + 2 \alpha < -Z+ 2 \lambda \le 0$, the right hand side
operator has infinitely many negative eigenvalues of finite
multiplicities. Therefore, so has $H_{\rho_\gamma}$. Eventually,
$\epsilon_{\rm F} < 0$ and 
$$
\gamma = \sum_{i=1}^{n} |\phi_i\rangle \langle \phi_i|  +
\sum_{i=n+1}^{m} n_i |\phi_i\rangle \langle \phi_i|
$$
where $0 \le n_i \le 1$ and where 
$$
-\frac 1 2 \Delta \phi_i+ V \phi_i + \left( \rho_\gamma \star
|\br|^{-1} \right) \phi_i + g'(\rho_\gamma) \, \phi_i =
\epsilon_i  \phi_i
$$
$\epsilon_1 < \epsilon_2 \le \epsilon_3 \le \cdots < 0$ denoting the
negative eigenvalues of $H_{\rho_\gamma}$ including multiplicities (by
standard arguments the ground state eigenvalue of $H_{\rho_\gamma}$ is
non-degenerate). It then follows from elementary elliptic regularity
results that all the $\phi_i$, hence $\rho_\gamma$, are in $H^2(\RR^3)$
and therefore vanish at infinity. Using Lemma~\ref{lem:exp_decay}, all
the $\phi_i$ decay exponentially fast to zero at infinity.

\medskip

Let us now analyze more in details the sequence $(\gamma_{2,n})_{n \in \NN}$.
As it is a minimizing sequence for $I^\infty_{\lambda-\alpha}$,
$(\rho_{\gamma_{2,n}})_{n \in \NN}$ cannot vanish, so that there exists
$\eta > 0$, $R > 0$ and such for all $n \in \NN$, $\int_{y_n + B_R}
\rho_{\gamma_{2,n}} \ge \eta$ for some $y_n \in \RR^3$. Thus, the sequence
$(\tau_{y_n}\gamma_{2,n}\tau_{-y_n})_{n \in \NN}$ converges for the weak-$\ast$
topology of ${\cal H}$ to some $\gamma' \in {\cal K}$ satisfying
$\tr(\gamma') \ge \eta > 0$. Let $\beta = \tr(\gamma')$.
Reasoning as
above, one can easily check that $\gamma'$ is a minimizer for
$I^\infty_{\beta}$, and that $I_\lambda = I_\alpha + I^\infty_\beta +
I^\infty_{\lambda-\alpha-\beta}$. Besides,
$$
\gamma' = 1_{(-\infty,\epsilon_{\rm F}')}(H^\infty_{\rho_{\gamma'}})+
\delta'
$$
where
$$
H^\infty_{\rho_{\gamma'}} = -\frac 1 2 \Delta + \rho_{\gamma'} \star
|\br|^{-1}  + g'(\rho_{\gamma'}),
$$
and where $0 \le \delta' \le 1$, $\mbox{Ran}(\delta') \subset 
\mbox{Ker}(H^\infty_{\rho_{\gamma'}}-\epsilon_{\rm F}')$, and 
$\epsilon_{\rm F'} \le 0$. 

\medskip

Assume for a while that one can choose $\epsilon_{\rm F'} < 0$. Then 
$$
\gamma' = \sum_{i=1}^{n'} |\phi_i'\rangle \langle \phi_i'|  +
\sum_{i=n'+1}^{m'} n_i' |\phi_i'\rangle \langle \phi_i'|,
$$
all the $\phi_i$'s being in $C^\infty(\RR^3)$ and decaying exponentially
fast at infinity. For $n \in \NN$ large
enough, the operator 
$$
\gamma_n = \min \left( 1 ,\|\gamma +
  \tau_{n{\bf e}}\gamma'\tau_{-n{\bf e}}\|^{-1} \right) 
(\gamma + \tau_{n{\bf e}}\gamma'\tau_{-n{\bf e}})
$$
then is in ${\cal K}$ and $\tr(\gamma_n) \le (\alpha+\beta)$. 
As both the $\phi_i$'s and the $\phi_i'$'s decay
exponentially fast to zero, a simple calculation shows that there exists
some $\delta > 0$ such that for $n$ large enough 
$$
{\cal E}^{\rm LDA}(\gamma_n) = {\cal E}^{\rm LDA}(\gamma) + {\cal
  E}^{{\rm LDA},\infty}(\gamma')  
- \frac{2\alpha(Z-2\beta)}{n} + O(e^{-\delta n}) = I^\alpha + I^\infty_\beta
- \frac{2\alpha(Z-2\beta)}{n} + O(e^{-\delta n}). 
$$
Hence, for $n$ large enough 
$$
I_{\alpha + \beta} \le I_{\tr(\gamma_n)} \le {\cal E}^{\rm LDA}(\gamma_n) <
I^\alpha + I^\infty_\beta . 
$$
Adding $I^\infty_{\lambda-\alpha-\beta}$ to both sides, we obtain that
$$
I_\lambda \le  I_{\alpha + \beta} + I^\infty_{\lambda-\alpha-\beta}
< I^\alpha + I^\infty_\beta + I^\infty_{\lambda-\alpha-\beta},
$$
which obviously contradicts the previously established equality
$I_\lambda = I_\alpha + I^\infty_\beta + I^\infty_{\lambda-\alpha-\beta}$.

\medskip 

It remains to exclude the case when $\epsilon_{\rm F'}$ has to be chosen
equal to zero. In this case, $0$ is an eigenvalue of
$H_{\rho_{\gamma'}}^\infty$ and there exists $\psi \in
\mbox{Ker}(H_{\rho_{\gamma'}}^\infty) \subset H^2(\RR^3)$ such that
$\|\psi\|_{L^2} = 1$ and $\gamma'\psi = \mu \psi$ with $\mu > 0$. We
then define for $0 < \eta < \mu$ and $n \in \NN$, 
\begin{eqnarray*}
\gamma_{n,\eta} & = &\min \left( 1 ,\|\gamma + \eta
  |\phi_{m+1}\rangle\langle\phi_{m+1}| +
  \tau_{n{\bf e}}(\gamma'-\eta|\psi\rangle\langle\psi|)\tau_{-n{\bf e}}\|^{-1}
\right) \\
& &  \qquad
(\gamma + \eta
  |\phi_{m+1}\rangle\langle\phi_{m+1}| +
  \tau_{n{\bf e}}(\gamma'-\eta|\psi\rangle\langle\psi|)\tau_{-n{\bf e}}).
\end{eqnarray*}
As $\gamma_{n,\eta}$ is in ${\cal K}$ and such that
$\tr(\gamma_{n,\eta}) \le \lambda$, it holds
$$
I_{\lambda} \le I_{\tr(\gamma_{n,\eta})} \le {\cal E}^{\rm LDA}(\gamma_{n,\eta}) .
$$
It is then easy to show that
$$
\lim_{n \to \infty}{\cal E}^{\rm LDA}(\gamma_{n,\eta}) = 
{\cal E}^{\rm LDA}(\gamma+ \eta |\phi_{m+1}\rangle\langle\phi_{m+1}|) 
+ {\cal E}^{^{\rm LDA},\infty}(\gamma'-\eta|\psi\rangle\langle\psi|).
$$
Besides, for $\eta> 0$ small enough
$$
{\cal E}^{\rm LDA}(\gamma+ \eta |\phi_{m+1}\rangle\langle\phi_{m+1}|) 
+ {\cal E}^{{\rm LDA},\infty}(\gamma'-\eta|\psi\rangle\langle\psi|) = 
{\cal E}^{\rm LDA}(\gamma) + {\cal E}^{{\rm LDA},\infty}(\gamma') + 2\eta
\epsilon_{m+1} + 
o(\eta).
$$
Reasoning as above, we obtain that for $\eta > 0$ small enough
$$
I_\lambda \le I_\lambda + 2\eta \epsilon_{m+1} +
o(\eta),
$$
which is in contradiction with the fact that $\epsilon_{m+1}$ is
negative. The proof is complete.

\subsection{Proof of Theorem~\ref{th:GGA_2e}}
\label{sec:GGA_2e}

For $\phi \in H^1(\RR^3)$, we set $\rho_\phi(x) = 2 |\phi(x)|^2$ and
$$
E(\phi) = \int_{\RR^3} |\nabla \phi|^2 + \int_{\RR^3} \rho_\phi V +
J(\rho_\phi) + E_{\rm xc}^{\rm GGA}(\rho_\phi).
$$
For all $\phi \in H^1(\RR^3)$ such that $\|\phi\|_{L^2} = 1$,
$\gamma_\phi = |\phi\rangle \langle \phi| \in {\cal K}_1$ and
${\cal E}(\gamma_\phi)=E(\phi)$. Therefore,
$$
I_1 \le \inf \left\{E(\phi), \; \phi \in H^1(\RR^3), \; \int_{\RR^3}
  |\phi|^2 = 1 \right\}.
$$
Conversely, for all $\gamma \in {\cal K}_1$, $\phi_\gamma =
\sqrt{\frac{\dps \rho_\gamma}2}$ satisfies $\phi_\gamma \in H^1(\RR^3)$,
$\|\phi\|_{L^2} = 1$ and
$$
{\cal E}^{\rm GGA}(\gamma) = {\cal E}^{\rm GGA}(|\phi_\gamma \rangle
\langle \phi_\gamma|) + 
\tr(-\Delta \gamma) - \frac 1 2 \int_{\RR^3} |\nabla
\sqrt{\rho_\gamma}|^2 \ge  {\cal E}^{\rm GGA}(|\phi_\gamma \rangle \langle
\phi_\gamma|) = E(\phi_\gamma). 
$$
Consequently, 
\begin{equation} \label{eq:EKS2e}
I_1 = \inf \left\{E(\phi), \; \phi \in H^1(\RR^3), \; \int_{\RR^3}
  |\phi|^2 = 1 \right\}
\end{equation}
and (\ref{eq:minRKS2}) has a minimizer for $N_p=1$, if and only if
(\ref{eq:EKS2e}) has a minimizer $\phi$ ($\gamma_\phi$ then is a
minimizer of (\ref{eq:minRKS2}) for $N_p=1$). We are therefore led to study
the minimization problem (\ref{eq:EKS2e}). In the GGA setting we are
interested in, $E(\phi)$ can be rewritten as
$$
E(\phi) = \int_{\RR^3} |\nabla \phi|^2 + \int_{\RR^3} \rho_\phi V +
J(\rho_\phi) + \int_{\RR^3} h(\rho_\phi,|\nabla \phi|^2).
$$
Conditions (\ref{eq:h0})-(\ref{eq:h0EL}) guarantee that $E$ is Fréchet
differentiable on $H^1(\RR^3)$ (see \cite{these_Arnaud} for details) and
that for all $(\phi,w) \in H^1(\RR^3) 
\times H^1(\RR^3)$,
$$
E'(\phi)\cdot w = 2\bigg( \frac 12 \int_{\RR^3} \left( 1 +
  \frac{\partial h}{\partial\kappa}\left(\rho_{\phi},|\nabla
    \phi|^2\right) \right)  \nabla \phi \cdot \nabla w + \int_{\RR^3} 
\left( V + \rho_{\phi} \star |\br|^{-1}  + 
 \frac{\partial h}{\partial \rho}\left(\rho_{\phi},|\nabla \phi|^2\right)
\right) \phi w \bigg).
$$
We now embed (\ref{eq:EKS2e}) in the family of problems 
\begin{equation} \label{eq:newPb}
J_\lambda = \inf \left\{E(\phi), \; \phi \in H^1(\RR^3), \;
  \int_{\RR^3} |\phi|^2 = \lambda \right\}
\end{equation}
and introduce the problem at infinity
\begin{equation} \label{eq:newPbinfty}
J_\lambda^\infty = \inf \left\{E^\infty(\phi), \; \phi \in H^1(\RR^3), \;
  \int_{\RR^3} |\phi|^2 = \lambda \right\}
\end{equation}
where
$$
E^\infty(\phi) = \int_{\RR^3} |\nabla \phi|^2 + 
J(\rho_\phi) + \int_{\RR^3} h(\rho_\phi,|\nabla \phi|^2).
$$
Note that reasoning as above, one can see that $J_\lambda = I_\lambda$
and $J_\lambda^\infty=I_\lambda^\infty$ for all $0 \le \lambda \le 1$
(while these equalities do not {\it a priori} hold true for $\lambda > 1$).

\medskip

The rest of this section consists in proving that (\ref{eq:newPb}) has a
minimizer for all $0 \le \lambda \le 1$. Let us start with a simple
lemma.

\medskip

\begin{lemma}\label{lem:preExistGGA} Let $0 \le \mu \le 1$ and let 
  $(\phi_n)_{n \in \NN}$ be a minimizing sequence for $J_\mu$
  (resp. for $J_\mu^\infty$) which converges to some $\phi \in
  H^1(\RR^3)$ weakly in $H^1(\RR^3)$. Assume that $\|\phi\|_{L^2}^2 =
  \mu$. Then $\phi$ is a minimizer for $J_\mu$ (resp. for $J_\mu^\infty$).
\end{lemma}

\medskip

\begin{proof}
Let $(\phi_n)_{n \in \NN}$ be a minimizing sequence for $J_\mu$ which
converges to $\phi$ weakly in $H^1(\RR^3)$.
For almost all $x \in \RR^3$, the function $z \mapsto |z|^2 +
h(\rho_\phi(x),|z|^2)$ is convex on $\RR^3$. Besides the function $t \mapsto
t+h(\rho_\phi(x),t)$ is Lipschitz on $\RR_+$,
uniformly in $x$. It follows that the functional
$$
\psi \mapsto \int_{\RR^3} \left( |\nabla \psi|^2 +
  h(\rho_\phi,|\nabla\psi|^2)\right)
$$  
is convex and continuous on $H^1(\RR^3)$. As $(\phi_n)_{n \in \NN}$
converges to $\phi$ weakly in $H^1(\RR^3)$, we get
$$
\int_{\RR^3} \left( |\nabla \phi|^2 +
  h(\rho_\phi,|\nabla\phi|^2)\right) \le \lim_{n \to \infty}
\int_{\RR^3} \left( |\nabla \phi_n|^2 +
  h(\rho_\phi,|\nabla\phi_n|^2)\right). 
$$
Besides, we deduce from (\ref{eq:h0infty}) that
$$
\left| \int_{\RR^3} \left( h(\rho_{\phi_n},|\nabla\phi_n|^2)  -
  h(\rho_\phi,|\nabla\phi_n|^2)\right) \right| \le 
C \|\phi_n -\phi\|_{L^2},
$$
where the constant $C$ only depends on $h$ and on the $H^1$ bound of
$(\phi_n)_{n \in \NN}$. As $(\phi_n)_{n \in \NN}$ converges to $\phi$ weakly in $L^2(\RR^3)$ and
as $\|\phi\|_{L^2} = \|\phi_n\|_{L^2}$ for all $n \in \NN$, the
convergence of $(\phi_n)_{n \in \NN}$ to $\phi$ holds strongly in
$L^2(\RR^3)$. Therefore,
\begin{eqnarray*}
\int_{\RR^3} |\nabla \phi|^2 + E_{\rm xc}^{\rm GGA}(\rho_{\phi}) & = & 
\int_{\RR^3} \left( |\nabla \phi|^2 +
  h(\rho_\phi,|\nabla\phi|^2)\right) \\ & \le &
\liminf_{n \to \infty} \int_{\RR^3} \left( |\nabla \phi_n|^2 +
  h(\rho_\phi,|\nabla\phi_n|^2)\right) \\ & & + \lim_{n \to \infty}
\int_{\RR^3} \left( h(\rho_{\phi_n},|\nabla\phi_n|^2)  -
  h(\rho_\phi,|\nabla\phi_n|^2) \right)
\\
& = & \liminf_{n \to \infty} 
\int_{\RR^3} |\nabla \phi_n|^2 + E_{\rm xc}^{\rm GGA}(\rho_{\phi_n}).
\end{eqnarray*}
Finally, as $(\phi_n)_{n \in \NN}$ is bounded in $H^1$ and converges
strongly to $\phi$ in $L^2(\RR^3)$, we infer that the convergence holds
strongly in $L^p(\RR^3)$ for all $2 \le p < 6$, yielding
$$
\lim_{n \to \infty} \int_{\RR^3} \rho_{\phi_n} V + J(\rho_{\phi_n}) =
\int_{\RR^3} \ \rho_{\phi} V + J(\rho_{\phi}). 
$$
Therefore,
$$
E(\phi) \le \liminf_{n \to \infty} E(\phi_n) = I_\mu.
$$
As $\|\phi\|_{L^2}^2=\mu$, $\phi$ is a minimizer for $J_\mu$. Obviously,
the same arguments can be applied to a minimizing sequence for $J_\mu^\infty$.
\end{proof}

\medskip

In order to prove that the minimizing sequences for $J_\lambda$ (or at
least some of 
them) are indeed precompact in $L^2(\RR^3)$, we will use  
the concentration-compactness method due to 
P.-L. Lions~\cite{PLL-CC}. Consider an Ekeland sequence $(\phi_n)_{n \in
  \NN}$ for (\ref{eq:newPb}), that is \cite{Ekeland} a sequence $(\phi_n)_{n \in
  \NN}$ such that
\begin{eqnarray}
&& \forall n \in \NN, \quad \phi_n \in H^1(\RR^3) \quad \mbox{and} \quad
\int_{\RR^3} \phi_n^2=\lambda \label{eq:phin1} \\ 
&& \lim_{n \to +\infty} E(\phi_n) = J_\lambda \label{eq:phin2} \\
&& \lim_{n \to +\infty} E'(\phi_n) + \theta_n \phi_n = 0 \quad \mbox{in
} H^{-1}(\RR^3) \label{eq:phin3}
\end{eqnarray}
for some sequence $(\theta_n)_{n \in \NN}$ of real numbers. As
on the one hand, $|\phi| \in H^1(\RR^3)$ and $E(|\phi|)=E(\phi)$ for
all $\phi \in H^1(\RR^3)$, and as on the other hand, the function
$\lambda \mapsto J_\lambda$ is decreasing on $[0,1]$, we can assume that 
\begin{equation} \label{eq:sign}
\forall n \in \NN, \quad \phi_n \ge 0 \mbox{ a.e. on } \RR^3 
\quad \mbox{and} \quad \theta_n \ge 0.
\end{equation}
Lastly, up to extracting subsequences, there is no restriction in
assuming the following convergences:
\begin{eqnarray}
& & \phi_n \rightharpoonup \phi \mbox{ weakly in } H^1(\RR^3),
\label{eq:phin5}  \\
& & \phi_n \rightarrow \phi \mbox{ strongly in } L^p_{\rm loc}(\RR^3)
\mbox{ for all } 2 \le p < 6 \label{eq:phin6} \\
& & \phi_n \rightarrow \phi \mbox{ a.e. in } \RR^3 \label{eq:phin7} \\
& & \theta_n \rightarrow \theta \mbox{ in } \RR, \label{eq:phin8}
\end{eqnarray}
and it follows from (\ref{eq:sign}) that $\phi \ge 0$ a.e. on $\RR^3$
and $\theta \ge 0$. Note that the Ekeland condition
(\ref{eq:phin3}) also reads
\begin{eqnarray} \nonumber 
\!\!\!\!\!\!\!\!\!\!
&& - \frac 1 2 \div \left( \left( 1 + \frac{\partial h}{\partial
      \kappa}\left(\rho_{\phi_n},|\nabla \phi_n|^2\right) 
\right) \nabla \phi_n 
\right) + \left( V + \rho_{\phi_n} \star |\br|^{-1} + 
 \frac{\partial h}{\partial \rho}\left(\rho_{\phi_n},|\nabla \phi_n|^2\right)
\right) \phi_n + \theta_n \phi_n \\
\!\!\!\!\!\!\!\!\!\!
&& \qquad \qquad = \eta_n \qquad \mbox{with} \qquad 
\eta_n \mathop{\longrightarrow}_{n \to 0} 0 \label{eq:Ekeland}
\mbox{ in } H^{-1}(\RR^3).
\end{eqnarray}
We can apply to the sequence $(\phi_n)_{n \in \NN}$ the following
version of the concentration-compactness lemma.

\medskip

\begin{lemma}[Concentration-compactness lemma~\cite{PLL-CC}] \label{lem:concom}
Let $\lambda > 0$ and 
$(\phi_n)_{n \in \NN}$ be a bounded sequence in $H^1(\mathbb{R}^3)$ such that 
$$ 
\forall n \in \NN, \quad \int_{\mathbb{R}^N} \phi_n^2 =\lambda.
$$
Then one can extract from $(\phi_n)_{n \in \NN}$ a subsequence
$(\phi_{n_k})_{k \in \NN}$ such that one of the following three conditions
holds true: \\
\begin{enumerate}
\item (Compactness) There exists a sequence $(y_k)_{k \in \NN}$  in
  $\mathbb{R}^3$, such that for all $\epsilon > 0$, there exists $R >0$
  such that    
$$ 
\forall k \in \NN, \quad \int_{y_k + B_R} \phi_{n_k}^{2}  \geq \lambda- \epsilon.
$$ 

\item (Vanishing) 
For all $R>0$,
$$ 
\lim_{k \to \infty} \sup_{y \in \mathbb{R}^3} \int_{y+B_R} \phi_{n_k}^{2} =
0.
$$

\item (Dichotomy) There exists $0 < \delta < \lambda$, such that for all
  $\epsilon >0$ there exists 
  \begin{itemize}
  \item a sequence $(y_k)_{k \in \NN}$ of points of $\RR^3$,
  \item a positive real number $R_1$ and a sequence of positive real numbers
    $(R_{2,k})_{k \in \NN}$ converging to $+\infty$,
  \item two sequences
  $(\phi_{1,k})_{k \in \NN}$ and $(\phi_{2,k})_{n \in \NN}$ bounded in
  $H^1(\mathbb{R}^3)$ (uniformly in $\epsilon$) 
  \end{itemize}
such that for all $k$: 
\begin{eqnarray*}
\left\{
\begin{aligned}
& \phi_{n_k} = \phi_{1,k} \quad \mbox{on } y_k+B_{R_1} \\
& \phi_{n_k} = \phi_{2,k} \quad \mbox{on } \RR^3 \setminus
(y_k+B_{R_{2,k}}) \\ 
&\left| \int_{\mathbb{R}^3} \phi_{1,k}^2-\delta \right| \leq
\epsilon, \quad
\left| \int_{\mathbb{R}^3} \phi_{2,k}^2-(\lambda-\delta) \right| \leq
\epsilon \\ 
&\lim_{k \to \infty} \mathrm{dist(Supp \; \phi_{1,k}, Supp \; \phi_{2,k})}
=  \infty \\ 
& \|\phi_{n_k}-\left(\phi_{1,k}+\phi_{2,k} \right)\|_{L^p(\RR^3)} \leq
C_p \, \epsilon^{\frac{6-p}{2p}} \quad \mbox{ for all } \; 2 \leq p < 6 \\ 
& \|\phi_{n_k} \|_{L^p(y_k+(B_{R_{2,k}}\setminus \overline{B}_{R_1}))} \leq
C_p \, \epsilon^{\frac{6-p}{2p}} \quad \mbox{ for all } \; 2 \leq p < 6 \\ 
&\liminf_{k \to \infty} \int_{\mathbb{R}^3} \left(\left|\nabla \phi_{n_k}
  \right|^2-\left|\nabla \phi_{1,k} \right|^2 - \left|\nabla \phi_{2,k}
  \right|^2 \right) \geq - C \epsilon, 
\end{aligned}
\right.
\end{eqnarray*}
\end{enumerate}
where the constants $C$ and $C_p$ only depend on the $H^1$ bound of
$(\phi_n)_{n \in \NN}$.
\end{lemma}

\medskip

\noindent
We then conclude using the following result.

\medskip

\begin{lemma} \label{lem:preTh2}
Let $(\phi_n)_{n \in \NN}$ satisfying
  (\ref{eq:phin1})-(\ref{eq:phin8}). Then using the terminology introduced in
  the concentration-compactness Lemma~\ref{lem:concom},
  \begin{enumerate}
  \item if some subsequence $(\phi_{n_k})_{k \in \NN}$ of $(\phi_{n})_{n
      \in \NN}$ satisfies the compactness condition, then
    $(\phi_{n_k})_{k \in \NN}$ converges to $\phi$ strongly in
    $L^p(\RR^3)$ for all $2 \le p < 6$~;
  \item a subsequence of $(\phi_n)_{n \in \NN}$ cannot vanish~;
  \item a subsequence of $(\phi_n)_{n \in \NN}$ cannot satisfy the
    dichotomy condition.
  \end{enumerate}
Consequently, $(\phi_n)_{n \in \NN}$ converges to
$\phi$ strongly in $L^p(\RR^3)$ for all $2 \le p < 6$. It follows that
$\phi$ is a minimizer to (\ref{eq:newPb}).
\end{lemma}

\medskip

\noindent
As the explicit form of the functions $\phi_{1,k}$ and $\phi_{2,k}$
arising in Lemma~\ref{lem:concom} will
be useful for proving the third assertion of Lemma~\ref{lem:preTh2}, we briefly
recall the proof of the former lemma.

\medskip

\begin{proof}[Sketch of the proof of Lemma~\ref{lem:concom}] 
The argument is based on the analysis of Levy's concentration function 
$$ 
Q_n(R) = \sup_{y \in \mathbb{R}^3} \int_{y+B_R} \phi_n^2.
$$
The sequence $(Q_n)_{n \in \NN}$ is a sequence of nondecreasing,
nonnegative, uniformly bounded functions such that
$\displaystyle{\lim_{R \rightarrow \infty} Q_n(R) = \lambda}$.\\   
There exists consequently a subsequence  $(Q_{n_k})_{k \in \NN}$ and a
nondecreasing nonnegative function $Q$ such that $(Q_{n_k})_{k \in \NN}$
converges pointwise to $Q$. We obviously have 
$$
\lim_{R \rightarrow \infty} Q(R) = \delta \in [0,\lambda].
$$
The case $\delta = 0$ corresponds to vanishing, while $\delta = \lambda$
corresponds to compactness. We now consider more in details the case
when $0< \delta <\lambda$ (dichotomy).  
Let $\xi$, $\chi$ be in $C^\infty(\RR^3)$ and such that $0 \leq \xi,
\chi \leq 1$, $\xi(x)=1$ if $|x| \leq 1$, $\xi(x) =0$ if $|x| \geq 2$,
$\chi(x)=0$ if $|x| \leq 1$, $\chi(x) =1$ if $|x| \geq 2$, $\|\nabla
\chi\|_{L^\infty} \le 2$ and $\|\nabla \xi\|_{L^\infty} \le 2$. For $R
> 0$, we denote by $\xi_R(\cdot) = \xi\left(\frac{\cdot}R\right)$ and 
$\chi_R(\cdot) = \chi\left(\frac{\cdot}R\right)$.
Let $\epsilon > 0$ and $R_1 \geq \epsilon^{-1}$ large enough for
$Q(R_1)\geq \delta-\frac{\epsilon}{2}$ to hold. Then, up to getting rid of the first
terms of the sequence, we can assume that for all $k$, we have
$Q_{n_k}(R_1) \geq \delta-\epsilon$ and 
$Q_{n_k}(2R_1) \leq \delta+ \frac \epsilon 2$. Furthermore, there exists
$y_k \in \RR^3$ such that  
$$
Q_{n_k}(R_1) = \int_{y_k+B_{R_1}} \phi_{n_k}^2
$$
and we can choose a sequence $(R_k')_{k \in \NN}$ of
positive real numbers greater than $R_1$, converging to infinity, such that
$Q_{n_k}(2R_k') \leq \delta + \epsilon$ for all $k \in \NN$. 
Consider now 
$$
\phi_{1,k} = \xi_{R_1}(\cdot-y_k) \phi_{n_k} \quad \mbox{and} \quad
\phi_{2,k} = \chi_{R_k'}(\cdot-y_k) \phi_{n_k}.
$$
Denoting by $R_{2,k} = 2 R'_k$, we clearly have 
$$ \left|\int_{\mathbb{R}^3} \phi_{1,k}^2 - \delta \right| \leq
\epsilon, \quad \left|\int_{\mathbb{R}^3} \phi_{2,k}^2 -
  (\lambda-\delta) \right| \leq \epsilon,$$
$$
\int_{y_k+(B_{R_{2,k}}\setminus \overline{B}_{R_1})} \phi_{n_k}^2 = 
\int_{R_1 < |\cdot -y_k| < R_{2,k}} \phi_{n_k}^2 
\leq Q_{n_k}(R_{2,k})-Q_{n_k}(R_1) \le 2\epsilon,
$$
and
\begin{eqnarray*}
\begin{aligned}
\int_{\mathbb{R}^3} \left|\phi_{n_k} -  (\phi_{1,k}+\phi_{2,k})\right|^2
&\leq \int_{\mathbb{R}^3} |1- \xi_{R_1}(\cdot-y_k) -
\chi_{R_k'}(\cdot-y_k)|^2 \phi_{n_k}^2 \\ 
&\leq \int_{R_1 \leq |\cdot-y_k| \leq R_{2,k}} \phi_{n_k}^2 \leq 2 \epsilon.
\end{aligned}
\end{eqnarray*}
Similarly, by Hölder and Gagliardo-Nirenberg-Sobolev inequalities, we
have for all $k$ and $2\leq p < 6$ 
$$ 
\| \phi_{n_k} -  (\phi_{1,k}+\phi_{2,k}) \|_{L^p} \leq \| \phi_{n_k}
\|_{L^p(y_k+(B_{R_{2,k}}\setminus \overline{B}_{R_1}))} \le 
C_p \epsilon^{\frac {(6-p)}{2p}}
$$
where the constant $C_p$ only depends on $p$ and on the $H^1$ bound on
$(\phi_n)_{n \in \NN}$. Finally, we have 
$\|\nabla  \xi_{R_1}\|_{L^{\infty}} \leq 2R_1^{-1} \leq 2 \epsilon$ and
$\|\nabla  \chi_{R_k'}\|_{L^{\infty}} \leq 2 (R_k')^{-1}  \leq 2\epsilon$,
so that  
$$ 
\left|\int_{\mathbb{R}^3} |\nabla \phi_{1,k}|^2 - \xi_{R_1}^2(\cdot-y_k)
  |\nabla \phi_{n_k}|^2 \right| \leq C \frac{\epsilon}2
$$ 
and
$$ 
\left|\int_{\mathbb{R}^3} |\nabla \phi_{2,k}|^2-
  \chi_{R_k'}^2(\cdot-y_k) |\nabla \phi_{n_k}|^2
\right| \leq C \frac{\epsilon}2
$$
where the constant $C$ only depend on the $H^1$ bound on $(\phi_n)_{n
  \in \NN}$.
Thus 
\begin{eqnarray*}
\int_{\mathbb{R}^3} |\nabla \phi_{n_k}|^2-|\nabla \phi_{1,k}|^2-|\nabla
\phi_{2,k}|^2 & \geq & \int_{\mathbb{R}^3} (1-
\xi_{R_1}^2(\cdot-y_k)-\chi_{R_k'}^2(\cdot-y_k)) |\nabla \phi_{n_k}|^2 -
C \epsilon \\ 
&\geq & - C \epsilon.
\end{eqnarray*}
\end{proof}

\medskip

\begin{proof}[Proof of the first two assertions of Lemma~\ref{lem:preTh2}]
Assume that there exists a sequence $(y_k)_{k
  \in \NN}$ in $\mathbb{R}^3$, such that for all $\epsilon > 0$, there
exists $R >0$  such that    
$$ 
\forall k \in \NN, \quad \int_{y_k + B_R} \phi_{n_k}^{2}  \geq \lambda- \epsilon.
$$ 
Two situations may be encountered: either $(y_k)_{k \in \NN}$ has a
converging subsequence, or $\dps \lim_{k \to \infty} |y_k| = \infty$. In
the latter case, we would have $\phi=0$, and therefore
$$
\lim_{k \to \infty} \int_{\RR^3} \phi_{n_k}^2 V = 0.
$$
Hence
$$
I_\lambda^\infty \le \lim_{k \to \infty} E^\infty(\phi_{n_k}) = \lim_{k
  \to \infty} E(\phi_{n_k}) = I_\lambda,
$$
which is in contradiction with the first assertion of
Lemma~(\ref{lem:I_lambda}). Therefore, $(y_k)_{k \in \NN}$ has a 
converging subsequence. It is then 
easy to see, using the strong convergence of $(\phi_n)_{n \in \NN}$
to $\phi$ in $L^2_{\rm loc}(\RR^3)$, that 
$$
\int_{\RR^3} \phi^2 \ge \int_{y+B_R} \phi^2 \ge \lambda-\epsilon,
$$
where $y$ is the limit of some converging subsequence of $(y_k)_{k \in
  \NN}$. This implies that $\|\phi\|_{L^2}^2 = \lambda$, hence
that $(\phi_n)_{n \in \NN}$ converges to $\phi$ strongly in
$L^2(\RR^3)$. As $(\phi_n)_{n \in \NN}$ is bounded in $H^1(\RR^3)$, this
convergence holds strongly in $L^p(\RR^3)$ for all $2 \le p < 6$.

\medskip

\noindent
Assume now that $(\phi_{n_k})_{k \in \NN}$ is vanishing. Then we would
have $\phi=0$, an eventuality that has already been excluded.
\end{proof}

\medskip

\begin{proof}[Proof of the third assertion of Lemma~\ref{lem:preTh2}]
Replacing $(\phi_n)_{n \in \NN}$ with a subsequence and using a
diagonal extraction argument, we can
assume that in addition to (\ref{eq:phin1})-(\ref{eq:phin8}), there
exists 
\begin{itemize}
\item a sequence $(y_n)_{n \in \NN}$ of points in $\RR^3$,
\item two increasing sequences of positive real numbers $(R_{1,n})_{n \in \NN}$
and $(R_{2,n})_{n \in \NN}$ such that
$$
\lim_{n\to\infty} R_{1,n} = \infty \quad \mbox{and} \quad 
\lim_{n\to\infty} R_{2,n}-R_{1,n} = \infty
$$
\item two sequences  $(\phi_{1,n})_{n \in \NN}$ and $(\phi_{2,n})_{n \in
    \NN}$  bounded in $H^1(\RR^3)$
\end{itemize}
such that
\begin{eqnarray*}
\left\{
\begin{aligned}
& \phi_{n} = \phi_{1,n} \quad \mbox{on } y_n+B_{R_{1,n}} \\
& \phi_{n} = \phi_{2,n} \quad \mbox{on } \RR^3 \setminus
(y_n+B_{R_{2,n}}) \\ 
& \lim_{n \to \infty} \int_{\mathbb{R}^3} \phi_{1,n}^2 = \delta, \quad
\lim_{n \to \infty} \int_{\mathbb{R}^3} \phi_{2,n}^2 = \lambda -\delta \\
&  \lim_{n \to \infty} \|\phi_n-(\phi_{1,n}+\phi_{2,n})\|_{L^p(\RR^3)} = 0
\quad \mbox{for all }  2 \leq p < 6 \\ 
& \lim_{n \to \infty}  \|\phi_{n} \|_{L^p(y_n+(B_{R_{2,n}}\setminus
  \overline{B}_{R_{1,n}}))} = 0 \quad \mbox{ for all } \; 2 \leq p < 6 \\ 
&\lim_{n \to \infty} \mathrm{dist(Supp \; \phi_{1,n}, Supp \;
  \phi_{2,n})} = \infty \\ 
&\liminf_{n \to \infty} \int_{\mathbb{R}^3} \left(\left|\nabla \phi_{n}
  \right|^2-\left|\nabla \phi_{1,n} \right|^2 - \left|\nabla \phi_{2,n}
  \right|^2 \right) \geq 0.  
\end{aligned}
\right.
\end{eqnarray*}
Besides, it follows from the construction of the functions $\phi_{1,n}$
and $\phi_{2,n}$ that
\begin{equation} \label{eq:sgnphijn}
\forall n \in \NN, \quad \phi_{1,n} \ge 0 \quad \mbox{and} \quad
\phi_{2,n} \ge 0 \quad \mbox{a.e. on } \RR^3.
\end{equation}
A straightforward calculation leads to
\begin{eqnarray}
E(\phi_n) & = & E^\infty(\phi_{1,n}) + \int_{\RR^3} \rho_{\phi_{1,n}} V
+ E^\infty(\phi_{2,n}) +  \int_{\RR^3} \rho_{\phi_{2,n}} V \nonumber \\
& & + \int_{\mathbb{R}^3} \left(\left|\nabla \phi_{n}
  \right|^2-\left|\nabla \phi_{1,n} \right|^2 - \left|\nabla \phi_{2,n}
  \right|^2 \right)  
+ \int_{\RR^3} \widetilde \rho_n V \nonumber  \\ 
& & +  D(\rho_{\phi_{1,n}},\rho_{\phi_{2,n}}) + 
D(\widetilde \rho_n,\rho_{\phi_{1,n}}+\rho_{\phi_{2,n}}) + \frac 1 2 D(\widetilde
\rho_n,\widetilde \rho_n) \nonumber  \\ 
& & + \int_{\RR^3}
(h(\rho_{\phi_n},|\nabla \phi_{n}|^2) - h(\rho_{\phi_{1,n}},|\nabla
\phi_{1,n}|^2) -   h(\rho_{\phi_{2,n}},|\nabla
\phi_{2,n}|^2)), \label{eq:ineqEphin}
\end{eqnarray}
where we have denoted by $\widetilde \rho_n = \rho_n -
\rho_{\phi_{1,n}} - \rho_{\phi_{2,n}}$. As
$$
|\widetilde \rho_n| \le 2 \chi_{y_n+(B_{R_{2,n}} \setminus
  \overline{B}_{R_{1,n}})} \, |\phi_n|^2,
$$
the sequence $(\widetilde\rho_n)_{n \in \NN}$ goes to zero in
$L^p(\RR^3)$ for all $1 \le p < 3$, yielding 
$$
\int_{\RR^3} \widetilde \rho_n V + D(\widetilde
\rho_n,\rho_{\phi_{1,n}}+\rho_{\phi_{2,n}}) + \frac 1 2 D(\widetilde
\rho_n,\widetilde \rho_n)  \mathop{\longrightarrow}_{n \to \infty} 0.
$$ 
Besides,
$$
D(\rho_{\phi_{1,n}},\rho_{\phi_{2,n}}) \le 4 \;  \mathrm{dist(Supp \;
  \phi_{1,n}, Supp \; 
  \phi_{2,n})}^{-1} \, \|\phi_{1,n}\|_{L^2}^2 \, \|\phi_{2,n}\|_{L^2}^2
\mathop{\longrightarrow}_{n \to \infty} 0
$$
and 
\begin{eqnarray*}
& & \left| \int_{\RR^3}(h(\rho_{\phi_n},|\nabla \phi_{n}|^2) -
h(\rho_{\phi_{1,n}},|\nabla \phi_{1,n}|^2) -
h(\rho_{\phi_{2,n}},|\nabla \phi_{2,n}|^2)) \right| \\ & & \le
\int_{y_n+(B_{R_{2,n}} \setminus \overline{B}_{R_{1,n}})} \left|
h(\rho_{\phi_n},|\nabla \phi_{n}|^2) \right| + \left| h(\rho_{\phi_{1,n}},|\nabla
\phi_{1,n}|^2) \right| + \left|   h(\rho_{\phi_{2,n}},|\nabla
\phi_{2,n}|^2) \right| \\ & & \le C
\left(\|\rho_{\phi_n}\|_{L^{p_-}(y_n+(B_{R_{2,n}} \setminus
    \overline{B}_{R_{1,n}}))}^{p_-} + \|\rho_{\phi_n}\|_{L^{p_+}(y_n+(B_{R_{2,n}}
  \setminus \overline{B}_{R_{1,n}}))}^{p_+} \right) \mathop{\longrightarrow}_{n
\to \infty} 0 
\end{eqnarray*}
(recall that $1 < p_\pm=1+\beta_\pm < \frac 53$). Lastly, 
as $\lim_{n \to \infty} \mathrm{dist(Supp \; \phi_{1,n}, Supp \;
  \phi_{2,n})} = \infty$, 
$$
\min \left( \left|\int_{\RR^3} \rho_{\phi_{1,n}} V \right|
,  \left|\int_{\RR^3} \rho_{\phi_{2,n}} V \right| \right) 
 \mathop{\longrightarrow}_{n \to \infty} 0. 
$$
It therefore follows from (\ref{eq:ineqEphin}) and from the continuity
of the functions $\lambda \mapsto J_\lambda$ and $\lambda \mapsto
J_\lambda^\infty$ that at least one of the inequalities below holds true 
\begin{equation} \label{eq:Jl1}
J_\lambda \ge J_\delta + J_{\lambda-\delta}^\infty \quad \mbox{(case 1)}
\quad \mbox{or} \quad 
J_\lambda \ge J_\delta^\infty + J_{\lambda-\delta} \quad \mbox{(case 2)}. 
\end{equation}
As the opposite inequalities are always satisfied, we obtain
\begin{equation} \label{eq:Jl2}
J_\lambda = J_\delta + J_{\lambda-\delta}^\infty \quad \mbox{(case 1)} 
\quad \mbox{or} \quad 
J_\lambda = J_\delta^\infty + J_{\lambda-\delta} \quad \mbox{(case 2)} 
\end{equation}
and that (still up to extraction)
\begin{equation} \label{eq:Jl3}
\left\{ \begin{array}{l}
\dps \lim_{n \to \infty} E(\phi_{1,n}) = J_\delta \\
\dps \lim_{n \to \infty} E^\infty(\phi_{2,n}) = J_{\lambda-\delta}^\infty 
\end{array} \right.  \quad \mbox{(case 1)} 
\quad \mbox{or} \quad 
\left\{ \begin{array}{l}
\dps \lim_{n \to \infty} E^\infty(\phi_{1,n}) = J_\delta^\infty \\
\dps \lim_{n \to \infty} E(\phi_{2,n}) = J_{\lambda-\delta} 
\end{array} \right.  \quad \mbox{(case 2)}. 
\end{equation}
Let us now prove that the sequence $(\psi_n)_{n \in \NN}$, where $\psi_n
=  \phi_n - (\phi_{1,n}+\phi_{2,n})$, goes to zero in $H^1(\RR^3)$. For
convenience, we rewrite $\psi_n$ as $\psi_n = e_n \phi_n$ where $e_n
= 1 - \xi_{R_{1,n}}(\cdot-y_n) - \chi_{R_{2,n}/2}(\cdot-y_n)$ and 
Ekeland's condition (\ref{eq:Ekeland}) as
\begin{equation} \label{eq:Ekeland2}
- \div(a_n\nabla\phi_n) + V \phi_n + (\rho_{\phi_n} \star |\br|^{-1})
\phi_n + V_n^- \phi_n^{1+2\beta_-} + V_n^+ \phi_n^{1+2\beta_+} + \theta_n \phi_n = \eta_n
\end{equation}
where
$$
\left\{ \begin{array}{l}
\dps a_n = \frac 1 2 \left( 1+\frac{\partial h}{\partial
    \kappa}(\rho_{\phi_n},|\nabla \phi_n|^2) \right)  \\
\dps V_n^- =  2^{\beta_-} \rho_{\phi_n}^{-\beta_-} \frac{\partial h}{\partial
    \rho}(\rho_{\phi_n},|\nabla \phi_n|^2) \chi_{\rho_{\phi_n} \le 1} \\
\dps V_n^+ = 2^{\beta_+}  \rho_{\phi_n}^{-\beta_+} \frac{\partial h}{\partial
    \rho}(\rho_{\phi_n},|\nabla \phi_n|^2) \chi_{\rho_{\phi_n} > 1} .
\end{array} \right.
$$
The
sequence $(V \phi_n + (\rho_{\phi_n} \star |\br|^{-1})
\phi_n + V_n^- \phi_n^{1+2\beta_-} + V_n^+ \phi_n^{1+2\beta_+} +
\theta_n \phi_n )_{n \in \NN}$ is bounded in
$L^2(\RR^3)$, $(\eta_n)_{n \in \NN}$ goes to zero in $H^{-1}(\RR^3)$,
and the sequence $(\psi_n)_{n \in \NN}$ is bounded in $H^1(\RR^3)$ and
goes to zero in $L^2(\RR^3)$. We therefore infer from
(\ref{eq:Ekeland2}) that
$$
\int_{\RR^3} a_n \nabla \phi_n \cdot \nabla \psi_n  \mathop{\longrightarrow}_{n
\to \infty} 0. 
$$
Besides $\nabla \psi_n = e_n \nabla \phi_n + \phi_n \nabla e_n$ with $0
\le e_n \le 1$ and $\|\nabla e_n\|_{L^\infty} \to 0$. Thus
$$
\int_{\RR^3} a_n e_n |\nabla \phi_n|^2  \mathop{\longrightarrow}_{n
\to \infty} 0. 
$$
As 
\begin{equation} \label{eq:preHconv}
\dps 0 < \frac a 2 \le a_n = \frac 1 2 \left( 1+\frac{\partial h}{\partial
    \kappa}(\rho_{\phi_n},|\nabla \phi_n|^2) \right) \le \frac b 2 <
\infty \quad \mbox{a.e. on } \RR^3 
\end{equation}
and $0 \le e_n^2 \le e_n \le 1$, we finally obtain
$$
\int_{\RR^3} e_n^2 |\nabla \phi_n|^2  \mathop{\longrightarrow}_{n
\to \infty} 0,
$$
from which we conclude that $(\nabla \psi_n)_{n \in \NN}$ goes to zero
in $H^1(\RR^3)$. Plugging this information in (\ref{eq:Ekeland2}) and
using the fact that the supports of $\phi_{1,n}$ and $\phi_{2,n}$ are
disjoint and go far apart when $n$ goes to infinity, we obtain
\begin{eqnarray*}
&& \!\!\!\!\!\! 
- \div(a_n\nabla\phi_{1,n}) + V \phi_{1,n} + (\rho_{\phi_{1,n}} \star
|\br|^{-1}) \phi_{1,n} + V_n^- \phi_{1,n}^{1+2\beta_-} + V_n^+
\phi_{1,n}^{1+2\beta_+} + \theta_n \phi_{1,n} 
\mathop{\longrightarrow}_{n \to \infty}^{H^{-1}} 0 
\\
&& \!\!\!\!\!\!
- \div(a_n\nabla\phi_{2,n}) + V \phi_{2,n} + (\rho_{\phi_{2,n}} \star
|\br|^{-1}) \phi_{2,n} + V_n^- \phi_{2,n}^{1+2\beta_-}
+ V_n^+ \phi_{2,n}^{1+2\beta_+} + \theta_n \phi_{2,n}
\mathop{\longrightarrow}_{n \to \infty}^{H^{-1}} 0.
\end{eqnarray*}
We can now assume that the sequences $(\phi_{1,n})_{n \in
  \NN}$ and $(\phi_{2,n})_{n \in \NN}$, which are bounded in
$H^1(\RR^3)$, respectively converge to $u_1$ and $u_2$ weakly in
$H^1(\RR^3)$, strongly in $L^p_{\rm loc}(\RR^3)$ for all $2 \le p < 6$
and a.e. in $\RR^3$. In virtue of (\ref{eq:sgnphijn}), we also have $u_1
\ge 0$ and $u_2 \ge 0$ a.e. on $\RR^3$.
To pass to the limit in the above equations, we use
a H-convergence result proved in Appendix (Lemma~\ref{lem:Hconvergence}).
The sequence $(a_n)_{n \in \NN}$ satisfying (\ref{eq:preHconv}), there
exists $a_\infty \in L^\infty(\RR^3)$ such that $\frac a2 \le a_\infty \le
\frac{b^2}{2a}$ and  (up to extraction) $a_nI_3 \rightharpoonup_H a_\infty I_3$
(where $I_3$ is the rank-$3$ identity matrix). Besides, the sequence
$(V_n^\pm)_{n \in \NN}$ is bounded in $L^\infty(\RR^3)$, so
that there exists $V^\pm \in L^\infty(\RR^3)$, such that (up to
extraction) $(V_n^\pm)_{n \in \NN}$ 
converges to $V^\pm$  for the weak-$*$ topology of
$L^\infty(\RR^3)$. Hence for $j=1,2$ (and up to extraction) 
$$
\left\{
\begin{array}{l}
\dps V\phi_{j,n} \mathop{\longrightarrow}_{n \to \infty} Vu_j \quad \mbox{
  strongly in } H^{-1}(\RR^3) \\
\dps V_n^\pm \phi_{j,n}^{1+2\beta_\pm}  \mathop{\rightharpoonup}_{n \to
  \infty} V^\pm u_j^{1+2\beta_\pm} 
\quad \mbox{ weakly in } L^2_{\rm loc}(\RR^3) \\ 
\dps (\rho_{\phi_{j,n}} \star
|\br|^{-1}) \phi_{j,n} + \theta_n \phi_{j,n}
\mathop{\rightharpoonup}_{n \to \infty}  (\rho_{u_{j}} \star
|\br|^{-1}) u_{j} + \theta u_{j} \quad \mbox{ stronly in }
L^2_{\rm loc}(\RR^3) .
\end{array}\right.
$$
We end up with
\begin{eqnarray}
&& \!\!\!\!\!\! 
- \div(a_\infty \nabla u_1) + V u_1 + (\rho_{u_1} \star
|\br|^{-1}) u_1 + V^- u_1^{1+2\beta_-} + V^+ u_1^{1+2\beta_+}  + \theta
u_1 = 0 \label{eq:equ1}  \\ 
&&  \!\!\!\!\!\! 
- \div(a_\infty \nabla u_2) + V u_2 + (\rho_{u_2} \star
|\br|^{-1}) u_2 + V^- u_2^{1+2\beta_-} + V^+ u_2^{1+2\beta_+}  + \theta
u_2 = 0. \label{eq:equ2}  
\end{eqnarray}
By classical elliptic regularity arguments \cite{GT} (see also the proof
of Lemma~\ref{lem:exp_decay} below), both $u_1$ and $u_2$ are in
$C^{0,\alpha}(\RR^3)$ for some $0 < \alpha < 1$ and vanish at infinity. 
Besides, exactly one of the two functions $u_1$ and $u_2$ is different
from zero. Indeed, if both $u_1$ and $u_2$ were equal to zero, then we
would have $\phi=0$, hence
$$
J_\lambda = \lim_{n \to \infty} E(\phi_n) = \lim_{n \to \infty}
E^\infty(\phi_n) = J_\lambda^\infty,
$$
which is in contradiction with the first assertion of
Lemma~\ref{lem:I_lambda} (recall that $J_\lambda=I_\lambda$ and
$J_\lambda^\infty = I_\lambda^\infty$ for all $0 \le \lambda \le
1$). On the other hand, as $ 
\mathrm{dist(Supp \; \phi_{1,n}, Supp \;  \phi_{2,n})} \to \infty$, at
least one of the functions $u_1$ and $u_2$ is equal to zero. 

\medskip

We only consider here the case when $u_2 =0$,
corresponding to case~1 in
(\ref{eq:Jl1})-(\ref{eq:Jl3}), since the other case can be dealt with
the same arguments. A key point of the proof consists in noticing that
apply Lemma~\ref{lem:H_GGA} (proved in Appendix) to (\ref{eq:equ1})
(note that $W=V^- u_1^{\beta_-} + V^+ u_1^{\beta_+}$ is nonpositive and
goes to zero at infinity) yields 
\begin{equation} \label{eq:thetapositive}
\theta > 0. 
\end{equation}
Consider now the sequence $(\widetilde \phi_{1,n})_{n \in \NN}$ defined
by $\widetilde \phi_{1,n} = \delta^{\frac 12} \phi_{1,n}
\|\phi_{1,n}\|_{L^2}^{-1}$. It is easy to check that
$$
\left\{ \begin{array}{l}
\dps \forall n \in \NN, \quad \widetilde \phi_{1,n} \in H^1(\RR^3), \quad
\int_{\RR^3} \widetilde \phi_{1,n}^2=\delta \quad \mbox{and} \quad
\widetilde \phi_{1,n} \ge 0 \mbox{ a.e. on } \RR^3 \\  
\dps  \lim_{n \to +\infty} E(\widetilde \phi_{1,n}) = J_\delta  \\
 \dps - \div(a_{1,n} \nabla\widetilde\phi_{1,n}) + V \widetilde\phi_{1,n} +
(\rho_{\widetilde\phi_{1,n}} \star 
|\br|^{-1}) \widetilde\phi_{1,n} + V_{1,n}^-
\widetilde\phi_{1,n}^{1+2\beta_-} + V_{1,n}^+
\widetilde\phi_{1,n}^{1+2\beta_+} + 
\theta_n \widetilde\phi_{1,n} 
\mathop{\longrightarrow}_{n \to \infty}^{H^{-1}} 0  \\
\dps (\widetilde\phi_{1,n})_{n \in \NN} \mbox{ converges to } \widetilde
v_1 \neq 0
\mbox{ weakly in $H^1$, strongly in $L^p_{\rm loc}$ for $2 \le p < 6$ and a.e. on
} \RR^3 
\end{array} \right.
$$
(with in fact $v_1 = \phi$). Likewise, the
sequence $((\lambda-\delta)^{\frac 12}  \,
\|\phi_{2,n}\|_{L^2}^{-1} \, \phi_{2,n})_{n \in \NN}$ being a minimizing
sequence for 
$J_{\lambda-\delta}^\infty$, it cannot vanish. Therefore, there exists $\gamma
> 0$, $R > 0$ and a sequence $(x_n)_{n \in \NN}$ of points of $\RR^3$
such that $\int_{x_n+B_R} |\phi_{2,n}|^2 \ge \gamma$. Then, denoting by
$\widetilde \phi_{2,n} = (\lambda-\delta)^{\frac 12} 
\, \|\phi_{2,n}\|_{L^2}^{-1} \, \phi_{2,n}(\cdot-x_n)$,
$$
\left\{ \begin{array}{l}
\dps \forall n \in \NN, \quad \widetilde \phi_{2,n} \in H^1(\RR^3), \quad
\int_{\RR^3} \widetilde \phi_{2,n}^2=\lambda-\delta \quad \mbox{and} \quad
\widetilde \phi_{2,n} \ge 0 \mbox{ a.e. on } \RR^3 \\  
\dps  \lim_{n \to +\infty} E^\infty(\widetilde \phi_{2,n}) =
J_{\lambda-\delta}^\infty  \\
 \dps - \div(a_{2,n}\nabla\widetilde\phi_{2,n}) +
(\rho_{\widetilde\phi_{2,n}} \star 
|\br|^{-1}) \widetilde\phi_{2,n} + V_{2,n}^-
\widetilde\phi_{2,n}^{1+2\beta_-} + V_{2,n}^+
\widetilde\phi_{2,n}^{1+2\beta_+} + 
\theta_n \widetilde\phi_{2,n} 
\mathop{\longrightarrow}_{n \to \infty}^{H^{-1}} 0  \\
\dps (\widetilde\phi_{2,n})_{n \in \NN} \mbox{ converges to }
v_2 \neq 0
\mbox{ weakly in $H^1$, strongly in $L^p_{\rm loc}$ for $2 \le p < 6$ and a.e. on
} \RR^3.  
\end{array} \right.
$$
It is important to note that the sequence $(a_{j,n})_{n \in \NN}$ and
$(V_{j,n}^\pm)_{n \in \NN}$ are such that
$$
\frac a2 \le a_{j,n} \le \frac b2 \quad \mbox{and} \quad
\|V_{j,n}^\pm\|_{L^\infty} \le 2^{\beta_+} C,
$$
where the constants $a$, $b$ and $C$ are those arising in
(\ref{eq:h0infty}) and (\ref{eq:h0EL}).

\medskip

We can now apply the concentration-compactness lemma to $(\widetilde
\phi_{1,n})_{n \in \NN}$ and to $(\widetilde \phi_{2,n})_{n \in \NN}$. As 
 $(\widetilde \phi_{j,n})_{n \in \NN}$ does not vanish, either it is
 compact or it splits into subsequences that are either compact or
 split, and so on. The next step consists in showing that this process
 necessarily terminates after a finite number of iterations. By
 contradiction, assume that it is not the case. We could then construct
 by repeated applications of the concentration-compactness lemma
 (see~\cite{these_Arnaud} for details) an infinity of sequences
 $(\widetilde\psi_{k,n})_{n \in \NN}$, such that for all $k
 \in \NN$ 
$$
\left\{ \begin{array}{l}
\dps \forall n \in \NN, \quad \widetilde \psi_{k,n} \in H^1(\RR^3), \quad
\int_{\RR^3} \widetilde \psi_{k,n}^2=\delta_k \quad \mbox{and} \quad
\widetilde \psi_{k,n} \ge 0 \mbox{ a.e. on } \RR^3  \\  
 \dps - \div(\widetilde a_{k,n} \nabla\widetilde\psi_{k,n}) +
(\rho_{\widetilde\psi_{k,n}} \star 
|\br|^{-1}) \widetilde\psi_{k,n} + \widetilde V_{k,n}^-
\widetilde\psi_{k,n}^{1+2\beta_-} + \widetilde V_{k,n}^+
\widetilde\psi_{k,n}^{1+2\beta_+} + 
\theta_n \widetilde\psi_{k,n} 
\mathop{\longrightarrow}_{n \to \infty}^{H^{-1}} 0  \\
\dps (\widetilde\psi_{k,n})_{n \in \NN} \mbox{ converges to }
w_k \neq 0
\mbox{ weakly in $H^1$, strongly in $L^p_{\rm loc}$ for $2 \le p < 6$ and a.e. on
} \RR^3,   
\end{array} \right.
$$
with
\begin{equation} \label{eq:sumK}
\sum_{k \in \NN} \delta_k \le \lambda,
\end{equation}
and with for all $k \in \NN$, 
$$
\frac a2 \le \widetilde a_{k,n} \le \frac b2 \quad \mbox{and} \quad
\|\widetilde V_{k,n}^\pm\|_{L^\infty} \le  2^{\beta_+} C.
$$
Using Lemma~\ref{lem:Hconvergence} to pass to the limit with respect to $n$
in the equation satisfied by $\widetilde\psi_{k,n}$, we obtain
\begin{equation} \label{eq:eqonwk}
-\div( \widetilde a_k \nabla w_k) + (\rho_{w_{k}} \star 
|\br|^{-1}) w_{k} + \widetilde V_{k}^-
w_{k}^{1+2\beta_-} + \widetilde V_{k}^+
w_{k}^{1+2\beta_+} + \theta w_{k} = 0,
\end{equation}
with 
$$
\frac a2 \le \widetilde a_{k} \le \frac {b^2}{2a} \quad \mbox{and} \quad
\|\widetilde V_{k}^\pm\|_{L^\infty} \le  2^{\beta_+} C.
$$
Besides, we infer from (\ref{eq:sumK}) that
$\dps \sum_{k \in \NN} \|w_k\|_{L^2}^2 \le \lambda$, hence that
$$
\lim_{k \to \infty} \|w_k\|_{L^2} = 0.
$$
It then easily follows from (\ref{eq:eqonwk}) that 
$$
\lim_{k \to \infty} \|\div(a_k\nabla w_k)\|_{L^2} = 0.
$$
We can now make use of the elliptic regularity result \cite{GT} (see
also the proof of Lemma~\ref{lem:exp_decay}) stating that
there exists a constant $C$, depending only on the positive constants
$a$ and $b$, such that for all $u \in H^1(\RR^3)$ such that
$\div(\widetilde a_k\nabla u) \in L^2(\RR^3)$, $u \in L^\infty(\RR^3)$ and 
$$
\| u \|_{L^\infty} \le C \left( \|u\|_{L^2} + \|\div(\widetilde a_k\nabla
  u)\|_{L^2} \right)
$$
and obtain
$$
\lim_{k \to \infty} \|w_k\|_{L^\infty} = 0.
$$
Lastly, we deduce from (\ref{eq:eqonwk}) that
$$
\theta \|w_k\|_{L^2}^2 \le C \left( \|w_k\|_{L^\infty}^{2\beta_-}  
+ \|w_k\|_{L^\infty}^{2\beta_+} \right) \|w_k\|_{L^2}^2. 
$$
As $\|w_k\|_{L^2} > 0$ for all $k \in \NN$, we obtain that
$$
\theta \le  C \left( \|w_k\|_{L^\infty}^{2\beta_-}  
+ \|w_k\|_{L^\infty}^{2\beta_+} \right) \mathop{\longrightarrow}_{k \to
\infty} 0,
$$
which obviously contradicts (\ref{eq:thetapositive}). We therefore
conclude from this analysis that, if dichotomy occurs, $(\phi_n)_{n \in
  \NN}$ splits in a finite number of compact bits. We are now going to
prove that this cannot be.

\medskip

If this was the case, there would exist $\delta_1 > 0$ and $\delta_2 >
0$ such that $0 < \delta_1+\delta_2 \le \lambda$ and two sequences 
$(u_{1,n})_{n \in \NN}$ and $(u_{2,n})_{n \in \NN}$ such that
$$
\left\{ \begin{array}{l}
\dps \forall n \in \NN, \quad u_{1,n} \in H^1(\RR^3), \quad \int_{\RR^3}
|u_{1,n}|^2 = \delta_1, \quad u_1 \ge 0 \mbox{ a.e. on } \RR^3 \\
\dps \lim_{n \to \infty} E(u_{1,n}) = I_{\delta_1} \\
\dps 
- \div(\alpha_{1,n} \nabla u_{1,n}) + V  u_{1,n} +
(\rho_{u_{1,n}} \star 
|\br|^{-1}) u_{1,n} +  v_{1,n}^-
u_{1,n}^{1+2\beta_-} +  v_{1,n}^+
u_{1,n}^{1+2\beta_+} + 
\theta_n u_{1,n} 
\mathop{\longrightarrow}_{n \to \infty}^{H^{-1}} 0
\end{array} \right.
$$
and
$$
\left\{ \begin{array}{l}
\dps \forall n \in \NN, \quad u_{2,n} \in H^1(\RR^3), \quad \int_{\RR^3}
|u_{2,n}|^2 = \delta_2 , \quad u_2 \ge 0 \mbox{ a.e. on } \RR^3 \\
\dps \lim_{n \to \infty}  E^\infty(u_{2,n}) = I_{\delta_2} \\
\dps - \div(\alpha_{2,n} \nabla u_{2,n}) +
(\rho_{u_{2,n}} \star 
|\br|^{-1}) u_{2,n} +  v_{2,n}^-
u_{2,n}^{1+2\beta_-} +  v_{2,n}^+
u_{2,n}^{1+2\beta_+} + 
\theta_n u_{2,n} 
\mathop{\longrightarrow}_{n \to \infty}^{H^{-1}} 0
\end{array} \right. 
$$
and converging weakly in $H^1(\RR^3)$ to $u_1$ and $u_2$ respectively, with $\|u_1\|_{L^2} = \delta_1$ and $\|u_2\|_{L^2} = \delta_2$ (as the weak limit of $(\phi_n)_{n \in \NN}$ in $L^2(\RR^3)$ is nonzero, one bit stays at
finite distance from the nuclei). It then follows from
Lemma~\ref{lem:preExistGGA} that $u_1$ and $u_2$ are minimizers for
$J_{\delta_1}$ and $J_{\delta_2}^\infty$ respectively:
$$
E(u_1) = J_{\delta_1}, \quad \|u_1\|_{L^2}^2 = \delta_1, \quad 
E(u_2) = J_{\delta_2}^\infty, \quad \|u_2\|_{L^2}^2 = \delta_2.
$$ 
Letting $n$ go to
infinity in the equations satisfied by $u_{1,n}$ and $u_{2,n}$ we also
have 
\begin{equation} \label{eq:ELu1}
- \div (\alpha_1 \nabla u_1) + V u_1 + (\rho_{u_1} \star 
|\br|^{-1}) u_1 + v_1^-
u_1^{1+2\beta_-} + v_1^+
u_1^{1+2\beta_+} + 
\theta u_1 = 0 
\end{equation}
and
\begin{equation} \label{eq:ELu2}
- \div (\alpha_2 \nabla u_2) + (\rho_{u_2} \star 
|\br|^{-1}) u_2 + v_2^-
u_2^{1+2\beta_-} + v_2^+
u_2^{1+2\beta_+} + 
\theta u_2 = 0,
\end{equation}
with $\frac a 2 \le \alpha_j \le \frac{b^2}{2a}$ and $\|
v_j^\pm \|_{L^\infty} \le 2^{\beta_+}C$. This shows in particular that $u_1$ and
$u_2$ are in $L^\infty(\RR^3)$. Applying Lemma~\ref{lem:exp_decay}, we
then obtain that there exists $\gamma > 0$, $f_1 \in H^1(\RR^3)$, $f_2 \in
H^1(\RR^3)$, $g_1 \in (L^2(\RR^3))^3$ and $g_2 \in (L^2(\RR^3))^3$ such that
\begin{equation} \label{eq:expDu1u2}
u_1 = e^{-\gamma |\cdot|} f_1, \quad u_2 = e^{-\gamma |\cdot|} f_2,
\quad \nabla u_1 = e^{-\gamma |\cdot|} g_1, \quad \nabla u_2 =
e^{-\gamma |\cdot|} g_2. 
\end{equation}
In addition, as $u_1 \ge 0$ and $u_2 \ge 0$, we also have $f_1 \ge 0$
and $f_2 \ge 0$. Let ${\bf e}$ be a given unit vector of $\RR^3$. For $t > 0$,
we set
$$
w_t(\br) = \alpha_t \, (u_1(\br) + u_2(\br-t {\bf e})) \quad
\mbox{where} \quad \alpha_t = (\delta_1+\delta_2)^{\frac 12} \, \|u_1 +
u_2(\cdot-t{\bf e})\|_{L^2}^{-1}. 
$$
Obviously, $w_t \in H^1(\RR^3)$ and $\|w_t\|_{L^2} = \delta_1+\delta_2$,
so that 
\begin{equation} \label{eq:wtpg}
E(w_t) \ge J_{\delta_1+\delta_2}.
\end{equation}
Besides,
\begin{eqnarray*}
\|u_1 + u_2(\cdot-t{\bf e})\|_{L^2}^2 & = & \int_{\RR^3} u_1^2 + \int_{\RR^3}
u_2^2 + 2\int_{\RR^3} f_1(\br) \, f_2(\br-t{\bf e}) \, e^{-\gamma
  (|\br|+|\br-t{\bf e}|)} \, d\br \\
& = & \delta_1 + \delta_2 + 2  \int_{\RR^3} f_1(\br) \, f_2(\br-t{\bf
  e}) \, e^{-\gamma (|\br|+|\br-t{\bf e}|)} \, d\br \\
& = & \delta_1 + \delta_2 +  O(e^{-\gamma t}),  
\end{eqnarray*}
yielding
$$
\alpha_t =  1 + O(e^{-\gamma t}).
$$
Likewise, we have
\begin{eqnarray}
& & \dps \int_{\RR^3} |\nabla w_t|^2 = \int_{\RR^3} |\nabla u_1|^2 +
\int_{\RR^3} |\nabla u_2|^2 + O(e^{-\gamma t})  \label{eq:Tsmall} \\
& & \dps   \int_{\RR^3} V |w_t|^2 = \int_{\RR^3} V |u_1|^2 +
\int_{\RR^3} V |u_2(\cdot-t{\bf e})|^2 + O(e^{-\gamma t})
\label{eq:Vnesmall} \\
& & \dps D(\rho_{w_t},\rho_{w_t}) = D(\rho_{u_1},\rho_{u_1}) +
D(\rho_{u_2},\rho_{u_2}) + 2 D(\rho_{u_1},\rho_{u_2(\cdot-t{\bf e})}) +
O(e^{-\gamma t}) \label{eq:Veesmall}.
\end{eqnarray}
The exchange-correlation term can then be dealt with as
follows. Denoting by
$$
r_t = \rho_{w_t} -\rho_{u_1} -\rho_{u_2(\cdot-t{\bf e})} = 2(\alpha_t^2-1) (|u_1|^2 +|u_2(\cdot-t{\bf e})|^2) + 4 \alpha_t^2
u_1u_2(\cdot-t{\bf e})
$$ 
and
$$ 
s_t = |\nabla w_t|^2-|\nabla u_1|^2 -|\nabla u_2(\cdot-t{\bf e})|^2 = (\alpha_t^2-1) (|\nabla u_1|^2 +|\nabla u_2(\cdot-t{\bf e})|^2) +
2\alpha_t^2 \nabla u_1 \cdot \nabla u_2(\cdot-t{\bf e}),
$$
and using (\ref{eq:h0infty}), (\ref{eq:h0EL}), (\ref{eq:expDu1u2}) and
the fact that $u_1$ and $u_2$ are bounded in $L^\infty(\RR^3)$, we obtain
\begin{eqnarray*}
& & \left| \int_{\RR^3} h(\rho_{w_t},|\nabla w_t|^2) - 
h(\rho_{u_1},|\nabla u_1|^2) - h(\rho_{u_2(\cdot-t{\bf e})},|\nabla
u_2(\cdot-t{\bf e})|^2) \right| \\ 
& \le &   \int_{B_{\frac t2}} \left| h(\rho_{u_1} + \rho_{u_2(\cdot-t{\bf
      e})} + r_t,|\nabla u_1|^2 + |\nabla u_2(\cdot-t{\bf e})|^2 + s_t) - 
h(\rho_{u_1},|\nabla u_1|^2) \right|  \\ & + & 
 \int_{t{\bf e}+B_{\frac t2}} \left| h(\rho_{u_2(\cdot-t{\bf e})}+\rho_{u_1}+r_t,|\nabla
   u_2(\cdot-t{\bf e})|^2+|\nabla u_1|^2+s_t)  - 
h(\rho_{u_2(\cdot-t{\bf e})},|\nabla
   u_2(\cdot-t{\bf e})|^2) \right| \\  
& + &  \int_{B_{\frac t2}} \left| h(\rho_{u_2(\cdot-t{\bf e})},|\nabla 
u_2(\cdot-t{\bf e})|^2) \right| + \int_{t{\bf e}+B_{\frac t2}} 
 \left| h(\rho_{u_1},|\nabla u_1|^2) \right| \\ & + & 
\int_{\RR^3 \setminus \left(B_{\frac t2}\cup (t{\bf
        e}+B_{\frac t2})\right)} |h(\rho_{w_t},|\nabla w_t|^2)|+  
h(\rho_{u_1},|\nabla u_1|^2)|+|h(\rho_{u_2(\cdot-t{\bf e})},|\nabla
u_2(\cdot-t{\bf e})|^2)| = O(e^{-\gamma t}).
\end{eqnarray*}
Combining (\ref{eq:Tsmall})-(\ref{eq:Veesmall}) together with the
above inequality, we obtain 
$$
E(w_t) \le J_{\delta_1} + J_{\delta_2}^\infty +  \int_{\RR^3} V
|u_2(\cdot-t{\bf e})|^2 + D(\rho_{u_1},\rho_{u_2(\cdot-t{\bf e})}) +
O(e^{-\gamma t}).
$$
Next, using (\ref{eq:expDu1u2}), we get
\begin{eqnarray*}
 \int_{\RR^3} V \rho_{u_2(\cdot-t{\bf e})} +
 D(\rho_{u_1},\rho_{u_2(\cdot-t{\bf e})}) 
& = & - Z t^{-1} \int_{\RR^3} \rho_{u_2} + t^{-1} \int_{\RR^3} \rho_{u_1} \,
\int_{\RR^3} \rho_{u_2} + o(t^{-1}) \\ & = &- 2 \delta_2 (Z-2\delta_1) t^{-1} +
o(t^{-1}).  
\end{eqnarray*}
Finally,
$$
E(w_t) \le J_{\delta_1} + J_{\delta_2}^\infty - 2 \delta_2 (Z-2\delta_1) t^{-1} +
o(t^{-1}) \le J_{\delta_1+\delta_2} - 2 \delta_2 (Z-2\delta_1) t^{-1} +
o(t^{-1}) < J_{\delta_1+\delta_2} 
$$
for $t$ large enough, which contradicts (\ref{eq:wtpg}). 
\end{proof}

\medskip

\begin{proof}[End of the proof of Lemma~\ref{lem:preTh2}]
As a consequence of the concentration-compactness lemma and of the first
three assertions of Lemma~\ref{lem:preTh2}, the sequence $(\phi_n)_{n
  \in \NN}$ converges to $\phi$ weakly in $H^1(\RR^3)$ and strongly in
$L^p(\RR^3)$ 
for all $2 \le p < 6$. In particular,
$$
\int_{\RR^3} \phi^2 = \lim_{n\to\infty} \int_{\RR^3} \phi_n^2 = \lambda.
$$
It follows from Lemma~\ref{lem:preExistGGA} that $\phi$ is a minimizer to
(\ref{eq:newPb}). 
\end{proof}

\section*{Appendix}

In this appendix, we prove three technical lemmas, which we make use of in
the proof of Theorem~\ref{th:GGA_2e}. These lemmas are concerned with
second-order elliptic operators of the form $-\div(A \nabla \cdot)$. For
the sake of generality, we deal with the case when $A$ is a
matrix-valued function, although $A$ is a real-valued function in the
two-electron GGA model.

\medskip

For $\Omega$ an open subset of $\RR^3$ and $0 < \lambda \le \Lambda <
\infty$, we denote by $M(\lambda,\Lambda,\Omega)$ the closed convex
subset of $L^\infty(\Omega,\RR^{3\times 3})$ consisting of the matrix
fields $A \in L^\infty(\Omega,\RR^{3\times 3})$ such that for
all $\xi \in \RR^3$ and almost all $x \in \Omega$, 
$$
\lambda |\xi|² \leq A(x) \xi \cdot \xi \quad \mbox{and} \quad
|A(x) \xi | \leq \Lambda |\xi|.
$$
We also introduce the set $M^s(\lambda,\Lambda,\Omega)$
of the matrix fields $A \in M(\lambda,\Lambda,\Omega)$ such 
that $A(x)$ is symmetric for almost all $x \in \Omega$. Obviously,
$M^s(\lambda,\Lambda,\Omega)$ also is a closed convex subset of
$L^\infty(\Omega,\RR^{3\times 3})$. 

\medskip

The first lemma is a H-convergence result, in the same line as those
proved in the original article by Murat and Tartar~\cite{MuratTartar},
which allows to pass to the 
limit in the Ekeland condition~(\ref{eq:Ekeland}). Recall that a
sequence $(A_n)_{n \in \NN}$ of elements of $M(\lambda,\Lambda,\Omega)$
is said to H-converge to some $A \in M(\lambda',\Lambda',\Omega)$, which
is denoted by $A_n \rightharpoonup_{\rm H} A$, if for every $\omega
\subset \subset \Omega$ the following property holds :  
$\forall f \in H^{-1}(\omega)$, the sequence $(u_n)_{n \in \NN}$ of the
elements of $H^1_0(\omega)$ such that
$$
-\mathrm{div}(A_n \nabla u_n) = \left. f \right|_{\omega} \quad \mbox{in
} H^{-1}(\omega)
$$ 
satisfies
\begin{eqnarray*}
\left \{
\begin{aligned}
& u_n \rightharpoonup u \; \mbox{ weakly in }  H^1_{0}(\omega) \\
& A_n \nabla u_n \rightharpoonup A \nabla u \; \mbox{ weakly in }  L^2(\omega)
\end{aligned}
\right.
\end{eqnarray*}
where $u$ is the solution in $H^1_{0}(\omega)$ to $-\mathrm{div}(A \nabla
u) = \left. f \right|_{\omega}$. It is known \cite{MuratTartar} that from any
bounded sequence $(A_n)_{n \in \NN}$ in $M(\lambda,\Lambda,\Omega)$
(resp. in  $M(\lambda,\Lambda,\Omega)$) one can extract a subsequence
which H-converges to some $A \in
M(\lambda,\lambda^{-1}\Lambda^2,\Omega)$ (resp. to some 
$A \in M^s(\lambda,\lambda^{-1}\Lambda^2,\Omega)$).

\medskip

\begin{lemma} \label{lem:Hconvergence} Let $\Omega$ be an open subset of
  $\RR^3$, $0 < \lambda \le \Lambda < \infty$, $0 < \lambda' \le \Lambda'
  < \infty$,  and $(A_n)_{n \in \NN}$ a
  sequence of elements of $M(\lambda,\Lambda,\Omega)$ which
  H-converges to some $A \in M(\lambda',\Lambda',\Omega)$. Let
  $(u_n)_{n \in \NN}$, $(f_n)_{n \in \NN}$ and $(g_n)_{n \in \NN}$ be
  sequences of elements of $H^1(\Omega)$, $H^{-1}(\Omega)$ and
  $L^2(\Omega)$ respectively, and $u \in H^1(\Omega)$, $f \in
  H^{-1}(\Omega)$ and $g \in L^2(\Omega)$ such that
\begin{eqnarray*}
\left \{
\begin{aligned}
&-\mathrm{div}(A_n \nabla u_n) = f_n + g_n \; \mbox{ in } H^{-1}(\Omega)
\mbox{ for all $n \in \NN$}\\
&u_n \rightharpoonup u \; \mbox{weakly in} \; H^1(\Omega) \\
&f_n \rightarrow f \; \mbox{strongly in} \; H^{-1}(\Omega) \\
&g_n \rightharpoonup g \; \mbox{weakly in} \; L^{2}(\Omega).
\end{aligned}
\right.
\end{eqnarray*}
Then $-\div(A\nabla u) = f + g$ and $A_n \nabla u_n \rightharpoonup A
\nabla u$ weakly in $L^2(\Omega)$.
\end{lemma}

\medskip

The second lemma is an extension of \cite[Lemma~II.1]{Lions} and of a
classical result on the ground state of Schrödinger
operators~\cite{RS4}. Recall that
\begin{eqnarray*}
L^2(\RR^3)+L^\infty_\epsilon(\RR^3) & = & \bigg\{ {\cal W} \, | \, \forall
  \epsilon > 0, \; \exists ({\cal W}_2,{\cal W}_\infty) \in L^2(\RR^3)
  \times L^\infty(\RR^3) \; \mbox{ s.t. } \\
& & \qquad \qquad \qquad \; \|{\cal W}_\infty\|_{L^\infty} \le
  \epsilon, \; {\cal W}={\cal W}_2+{\cal W}_\infty\bigg\}. 
\end{eqnarray*} 

\medskip

\begin{lemma} \label{lem:H_GGA}
Let $0 < \lambda
  \le \Lambda < \infty$, $A \in M^s(\lambda,\Lambda,\RR^3)$, 
$W \in L^2(\RR^3)+L^\infty_\epsilon(\RR^3)$ such that
$W_+ = \max(0,W) \in
  L^2(\RR^3)+L^3(\RR^3)$ and $\mu$ a positive Radon measure on $\RR^3$
  such that $\mu(\RR^3) < Z=\sum_{k=1}^M z_k$. Then, 
$$
H = - \div(A \nabla \cdot) + V + \mu \star |\br|^{-1} + W
$$
defines a self-adjoint operator on $L^2(\RR^3)$
with domain 
$$
D(H) = \left\{ u \in H^1(\RR^3) \, | \, \div(A\nabla u) \in
  L^2(\RR^3) \right\}. 
$$
Besides,
$D(H)$ is dense in $H^1(\RR^3)$ and
included in $L^\infty(\RR^3) \cap C^{0,\alpha}(\RR^3)$ for some
$\alpha > 0$, and any function of $D(H)$ vanishes at infinity. 
In addition,
\begin{enumerate}
\item $H$ is bounded from below, $\sigma_{\rm ess}(H) \subset
  [0,\infty)$ and $H$ has an infinite number of negative eigenvalues;
\item the lowest eigenvalue $\mu_1$ of $H$ is simple and there exists
  an eigenvector $u_1 \in D(H)$ of $H$ associated with $\mu_1$ such that
  $u_1 > 0$ on $\RR^3$;
\item if $w \in D(H)$ is an eigenvector of $H$ such that $w \geq 0$ on
  $\mathbb{R}^3$, then there exists $\alpha > 0$ such that $w = \alpha
  u_1$. 
\end{enumerate}
\end{lemma}

\medskip

The third lemma is used to prove that the ground state density of the
GGA Kohn-Sham model exhibits exponential decay at infinity (at least
for the two electron model considered in this article). 

\medskip

\begin{lemma} \label{lem:exp_decay}
Let $0 < \lambda \le \Lambda < \infty$, $A \in
  M(\lambda,\Lambda,\RR^3)$, ${\cal V}$ a function of $L^{\frac 65}_{\rm
    loc}(\RR^3)$ which vanishes at infinity, $\theta > 0$ and $u \in
  H^1(\RR^3)$ such that
$$
- \mathrm{div}(A \nabla u) + {\cal V} u + \theta u = 0 \quad \mbox{in }
{\cal D}'(\mathbb{R}^3).  
$$
Then there exists $\gamma > 0$ depending on $(\lambda,\Lambda,\theta)$
such that $e^{\gamma |\br|}u \in H^1(\RR^3)$. 
\end{lemma}

\medskip

\begin{proof}[Proof of Lemma~\ref{lem:Hconvergence}] 
Let us denote by $\xi_n = A_n \nabla u_n$. One can extract from the
sequence $(\xi_n)_{n \in \NN}$, which is bounded in $L^2$, a subsequence
$(\xi_{n_k})_{k \in \NN}$ which converges weakly in
$L^2(\Omega)$ to some $\xi$ solution to $-\div(\xi) = f + g$ in
$H^{-1}(\Omega)$. The proof will be completed if we can show 
that we necessarily have $\xi = A \nabla u$. 
Consider $\omega \subset \subset \Omega$, $q \in H^{-1}(\omega)$ and
$v_n \in H^1_0(\omega)$ satisfying  
$$ 
-\mathrm{div}(A_n^{*} \nabla v_n) = q \quad \mbox{in } H^{-1}(\omega). 
$$
As the sequence $(A_n^{*})_{n \in \NN}$ $H$-converges to $A^{*}$ 
\cite{MuratTartar}, it holds
\begin{eqnarray*}
\left \{
\begin{aligned}
& v_n \rightharpoonup v  \; \mathrm{in} \; H^1_0(\omega) \\
& A_n^{*} \nabla v_n \rightharpoonup A^{*}\nabla v  \; \mathrm{in} \;
L^{2}(\omega)
\end{aligned}
\right.
\end{eqnarray*}
where $v$ is the solution to $-\div(A^* \nabla v) = q$ in $H^1_0(\omega)$.
Let $\phi \in C^\infty_c(\omega)$. As
\begin{eqnarray*}
\left \{
\begin{aligned}
&\phi v_n \rightharpoonup \phi v \; \mathrm{in} \; H^{1}_0(\omega)\\
&\phi v_n \rightarrow \phi v  \; \mathrm{in} \; L^{2}(\omega)\\
&\nabla\phi \, v_n \rightarrow \nabla \phi \, v  \; \mathrm{in} \;
(L^{2}(\omega))^3\\ 
&\nabla\phi \, u_n \rightarrow \nabla \phi \, u  \; \mathrm{in} \;
(L^{2}(\omega))^3, 
\end{aligned} 
\right.
\end{eqnarray*}
we have on the one hand
\begin{eqnarray*}
\int_{\omega} \xi_{n_k} \cdot \nabla v_{n_k} \phi &=& -(\mathrm{div
  \xi_{n_k}}, \phi v_{n_k})_{H^{-1}(\omega),H^1_0(\omega)} -
\int_{\omega} \xi_{n_k} \cdot \nabla \phi \, v_{n_k}  \\
&=& -(f_{n_k}, \phi v_{n_k})_{H^{-1}(\omega),H^1_0(\omega)} -
\int_{\omega} g_{n_k} \phi v_{n_k} - \int_{\omega} \xi_{n_k} \cdot
\nabla \phi \, v_{n_k} \\ 
&\rightarrow & 
-(f, \phi v)_{H^{-1}(\omega),H^1_0(\omega)} - \int_{\omega} g \phi v -
\int_{\omega} \xi \cdot \nabla \phi \, v \\ 
& = &  -(\mathrm{div \xi}, \phi v)_{H^{-1}(\omega),H^1_0(\omega)} -
\int_{\Omega} \xi \cdot \nabla \phi \, v = \int_{\Omega} \xi \cdot
\nabla v \, \phi,
\end{eqnarray*}
and on the other hand
\begin{eqnarray*}
\int_{\omega} \xi_{n_k} \cdot \nabla v_{n_k} \phi &=& 
\int_\omega \nabla u_{n_k} \cdot (A^*\nabla v_{n_k}) \phi \\
&=& - \int_\omega u_{n_k} (A^* \nabla v_{n_k}) \cdot \nabla \phi +
\int_\omega u_{n_k} q \phi \\ 
&\rightarrow & - \int_\omega u (A^* \nabla v) \cdot \nabla \phi +
\int_\omega u q \phi = \int_\omega \nabla u \cdot (A^* \nabla v) \phi = 
 \int_\omega (A \nabla u) \cdot \nabla v \, \phi. 
\end{eqnarray*}
Therefore,
$$
\int_{\omega} \xi \cdot \nabla v \phi = \int_\omega (A \nabla u) \cdot
\nabla v \phi. 
$$
As the above equality holds true for all $\omega$, all $v \in
H^1_0(\omega)$ and all
$\phi \in C^\infty_c(\omega)$, we finally obtain $\xi = A \nabla u$.
\end{proof}

\medskip

\begin{proof}[Proof of Lemma~\ref{lem:H_GGA}] 
The quadratic form $q_0$ on $L^2(\RR^3)$ with domain $D(q_0) =
H^1(\RR^3)$, defined by
$$
\forall (u,v) \in D(q_0) \times D(q_0), \quad 
q_0(u,v) = \int_{\RR^3} A \nabla u \cdot \nabla v,
$$
is symmetric and positive. It is also closed since the norm
$\sqrt{\|\cdot\|_{L^2}^2+q_0(\cdot)}$ is equivalent to the usual $H^1$
norm. This implies that $q_0$ is the quadratic form of a unique
self-adjoint operator $H_0$ on $L^2(\RR^3)$, whose domain $D(H_0)$ is
dense in $H^1(\RR^3)$. It is easy to check that $D(H_0) = \left\{u \in
  H^1(\RR^3) \; | \; \div(A\nabla u) \in L^2(\RR^3) \right\}$ and that
$$
\forall u \in D(H_0), \quad H_0u = -\div(A\nabla u).
$$ 
Using classical elliptic regularity results \cite{GT}, we obtain that
there exists two constants $0 < \alpha < 1$ and $C \in \RR_+$ (depending
on $\lambda$ and $\Lambda$) such that for all regular bounded domain
$\Omega \subset \subset \RR^3$, and all $v \in H^1(\Omega)$ such that
$\div(A\nabla v) \in L^2(\Omega)$, 
$$
\|v\|_{C^{0,\alpha}(\bar\Omega)}:= \sup_\Omega |v| + \sup_{(\br,\br')
  \in \Omega \times \Omega} \frac{|v(\br)-v(\br')|}{|\br-\br'|^\alpha} 
\le C \left(\|v\|_{L^2(\Omega)}+\|\div(A\nabla v)\|_{L^2(\Omega)} \right).  
$$
It follows that on the one hand, $D(H_0) \hookrightarrow
L^\infty(\RR^3) \cap C^{0,\alpha}(\RR^3)$, with
\begin{equation} \label{eq:boundDH0}
\forall u \in D(H_0), \quad \| u \|_{L^\infty(\RR^3)} +
\sup_{(\br,\br')
  \in \RR^3 \times \RR^3} \frac{|v(\br)-v(\br')|}{|\br-\br'|^\alpha} \le
C \left(\|u\|_{L^2}+\|H_0u\|_{L^2} \right), 
\end{equation}
and that on the other hand, any $u \in D(H_0)$ vanishes at infinity.

\medskip

Let us now prove that the multiplication by ${\cal W} = V + \mu \star
|\br|^{-1} + W$ defines a compact perturbation of $H_0$. For this
purpose, we consider a sequence $(u_n)_{n \in \NN}$ of elements of
$D(H_0)$ bounded for the norm $\|\cdot\|_{H_0} = (\|\cdot\|_{L^2}^2+
\|H_0\cdot\|_{L^2}^2)^{\frac 12}$. Up to
extracting a subsequence, we can assume without loss of generality that
there exists $u \in D(H_0)$ such that:
\begin{eqnarray*}
\left \{
\begin{aligned}
u_n &\rightharpoonup u \; \mathrm{in} \; H^1(\mathbb{R}^3) \; \mathrm{and} \; L^p(\mathbb{R}^3) \; \mathrm{for} \; 2 \leq p \leq 6\\
u_n &\rightarrow u \; \mathrm{in} \; L^p_{loc}(\mathbb{R}^3) \; \mathrm{with} \; 2 \leq p < 6\\
u_n &\rightarrow u \; a.e.
\end{aligned}
\right.
\end{eqnarray*}
Besides, it is then easy to check that the potential ${\cal W} = V + \mu
\star |\br|^{-1} + W$ belongs to $L^2 + L^\infty_\epsilon(\RR^3)$.
Let $\epsilon > 0$ and  $({\cal W}_2,{\cal W}_\infty) \in L^2(\RR^3)
  \times L^\infty(\RR^3)$ such that $\|{\cal W}_\infty\|_{L^\infty} \le
  \epsilon$ and ${\cal W}={\cal W}_2+{\cal W}_\infty$. On the one hand,
$$
\| {\cal W}_\infty (u_n-u) \|_{L^2} \le 2 \, \epsilon \, \sup_{n \in
  \NN} \|u_n\|_{H_0},  
$$
and on the other hand
$$
\lim_{n \to \infty} \| {\cal W}_2 (u_n-u) \|_{L^2} = 0.
$$
The latter result is obtained from Lebesgue's dominated convergence
theorem, using the fact that it follows from (\ref{eq:boundDH0}) that
$(u_n)_{n \in \NN}$ is bounded in $L^\infty(\RR^3)$. Consequently,
$$
\lim_{n \to \infty} \| {\cal W}u_n - {\cal W}u \|_{L^2} = 0,
$$
which proves that ${\cal W}$ is a $H_0$-compact operator. We can
therefore deduce from Weyl's theorem that $H = H_0 +{\cal W}$ defines a
self-adjoint operator on $L^2(\RR^3)$ with domain $D(H)=D(H_0)$, and
that $\sigma_{\rm ess}(H) = \sigma_{\rm ess}(H_0)$. As $q_0$ is
positive, $\sigma(H_0) \subset \RR_+$ and therefore 
$\sigma_{\rm ess}(H) \subset \RR_+$.

\medskip

Let us now prove that $H$ has an infinite number of negative eigenvalues
which forms an increasing sequence converging to zero.
First, $H$ is bounded below since for all $v \in D(H)$ such that
$\|v\|_{L^2} =1$,
\begin{eqnarray*}
\langle v |H|v \rangle & = & \int_{\RR^3} A\nabla v \cdot \nabla v +
\int_{\RR^3} {\cal W} v^2 \\
& \ge & \lambda \|\nabla v\|_{L^2}^2 - \|{\cal W}_2\|_{L^2} \|\nabla
v\|_{L^2}^{\frac 32} - \epsilon  \\
& \ge & - \frac{27}{256} \lambda^{-3} \|{\cal W}_2\|^4 - \epsilon. 
\end{eqnarray*}
In order to prove that $H$ has at least $N$ negative eigenvalues,
including multiplicities, we can proceed as in the proof of
\cite[Lemma~II.1]{Lions}. Let us indeed consider $N$ radial functions
$\phi_1$, ..., $\phi_N$
in ${\cal D}(\RR^3)$ such that for all $1 \le i \le
N$, $\mbox{supp}(\phi_i) \in B_{i+1} \setminus \overline{B}_i$ and
$\int_{\RR^3} |\phi_i|^2 = 1$. Denoting by $\phi_{i,\sigma}(\cdot) =
\sigma^{\frac 32} \phi_i(\sigma \cdot)$, we have
$$
\int_{\RR^3} A \nabla \phi_{i,\sigma} \cdot \nabla \phi_{i,\sigma}
\le \sigma^2 \Lambda \|\nabla \phi_i\|_{L^2}^2,
$$
and
\begin{eqnarray*}
\int_{\RR^3} W |\phi_{i,\sigma}|^2 & \le & 
\left( \int_{B_{(i+1)\sigma^{-1}} \setminus \overline{B}_{i\sigma^{-1}}}
W_2^2 \right)^{\frac 12} \|\phi_{i,\sigma} \|_{L^4}^2 
+ \left( \int_{B_{(i+1)\sigma^{-1}} \setminus \overline{B}_{i\sigma^{-1}}}
W_3^3 \right)^{\frac 12} \|\phi_{i,\sigma} \|_{L^3}^2 \\
& = &  \sigma^{\frac 32}
\left( \int_{B_{(i+1)\sigma^{-1}} \setminus \overline{B}_{i\sigma^{-1}}}
W_2^2 \right)^{\frac 12} \|\phi_{i} \|_{L^4}^2 
+ \sigma \left( \int_{B_{(i+1)\sigma^{-1}} \setminus \overline{B}_{i\sigma^{-1}}}
W_3^3 \right)^{\frac 12} \|\phi_{i} \|_{L^3}^2 \\
& = & o(\sigma)
\end{eqnarray*}
where we have split $W_+=\max(0,W)$ as $W_+=W_2+W_3$ with $W_2 \in
L^2(\RR^3)$ and $W_3 \in L^3(\RR^3)$.
Besides, we deduce from Gauss theorem that
{ $$
\int_{\RR^3} (\mu \star |\br|^{-1}) \phi_{i,\sigma}^2 =
\int_{\RR^3} \int_{\RR^3}
\frac{|\phi_{i,\sigma}(\br)|^2}{\max(|\br|,|\br'|)} \, d\br \, d\mu(\br') 
\le \sigma \, \mu(\RR^3) \, \int_{\RR^3} \frac{|\phi(\br)|^2}{|\br|} \,
d\br, 
$$}
and that, for $\sigma$ small enough,
$$
\int_{\RR^3} V |\phi_{i,\sigma}|^2 = \sigma \, Z \, \int_{\RR^3}
\frac{|\phi(\br)|^2}{|\br|} \, d\br. 
$$
Thus, 
\begin{eqnarray*}
\langle \phi_{i,\sigma}|H|\phi_{i,\sigma}\rangle & \le  & \sigma 
\left( \mu(\RR^3)-Z \right) \, \int_{\RR^3} \frac{|\phi(\br)|^2}{|\br|} \,
d\br + \mathop{o}_{\sigma \to 0} (\sigma),
\end{eqnarray*}
yielding $\langle \phi_{i,\sigma}|H|\phi_{i,\sigma}\rangle < 0$ for
$\sigma > 0$ small enough. As $\phi_{i,\sigma}$ and $\phi_{j,\sigma}$
have disjoint supports when $i \neq j$, we also have
$$
\max_{\phi \in \mbox{span}(\phi_{1,\sigma},\cdots,\phi_{N,\sigma}), \,
  \|\phi\|_{L^2} = 1} \langle \phi|H|\phi\rangle < 0
$$
for $\sigma > 0$ small enough. It follows from Courant-Fischer formula
\cite{RS4}
and from the fact that $\sigma_{\rm ess}(H) \subset \RR_+$ that $H$ has
at least $N$ negative eigenvalues, including multiplicites.

\medskip

The lowest eigenvalue of $H$, which we denote by $\mu_1$, is
characterized by
\begin{equation} \label{eq:charac_u1}
\mu_1 = \inf \left\{ \int_{\RR^3} A\nabla u \cdot \nabla u +
  \int_{\RR^3} {\cal W} |u|^2, \quad u \in H^1(\RR^3), \quad
  \|u\|_{L^2}=1 \right\},
\end{equation}
and the minimizers of (\ref{eq:charac_u1}) are exactly the set of the
normalized eigenvectors of $H$ associated with $\mu_1$. Let $u_1$ be a
minimizer (\ref{eq:charac_u1}). As for all $u \in H^1(\RR^3)$, $|u| \in
H^1(\RR^3)$ and $\nabla |u| = \mbox{sgn}(u) \nabla u$ a.e. on $\RR^3$,
$|u_1|$ also is a minimizer to (\ref{eq:charac_u1}). Up to replacing
$u_1$ with $|u_1|$, there is therefore no restriction in assuming that
$u_1 \ge 0$ on $\RR^3$. We thus have
$$
u_1 \in H^1(\RR^3) \cap C^0(\RR^3), \quad u_1 \ge 0 \quad \mbox{and}
\quad -\div(A\nabla u_1) + g u_1 = 0
$$
with $g = {\cal W}-\mu_1 \in L^p_{\rm loc}(\RR^3)$ for some $p > \frac
32$ (take $p=2$). A Harnack-type inequality due to
Stampacchia~\cite{Stampacchia} then implies that if $u_1$ has a zero in
$\RR^3$, then $u_1$ is identically zero. As $\|u_1\|_{L^2} =1$, we
therefore have $u_1 > 0$ on $\RR^3$.

\medskip

Consider now $w \in D(H) \setminus \left\{0\right\}$ such that $Hw=\mu
w$ and $w \ge 0$ on $\RR^3$. It holds
$$
\mu \int_{\RR^3} w_1 w = \langle w|H|w_1 \rangle = \mu_1 \int_{\RR^3} w_1 w.
$$
As $w$ is not identically equal to zero and as $w_1 > 0$ on $\RR^3$,
$\int_{\RR^3} w_1w > 0$, from which we deduce that $\mu=\mu_1$. It
remains to prove that $\mu_1$ is a non-degenerate eigenvalue. By
contradiction, let us assume that there exists $v \in D(H)$ such that 
$H v = \mu_1 v$, $\|v\|_{L^2} =1$ and $(v,u_1)_{L^2} = 0$. Reasoning as
above, $|v|$ also is an eigenvector of $H$ associated with $\mu_1$ and
$|v| > 0$ on $\RR^3$. Since $D(H) \subset C^0(\RR^3)$, $v$ is continuous
on $\RR^3$, so that either $v=|v|$ on $\RR^3$ or $v = -|v|$ on
$\RR^3$. In any case, $\left| \int_{\RR^3} u_1v \right| =
\int_{\RR^3} u_1|v| > 0$, which is in contradiction with the fact that
$(u_1,v)_{L^2} = 0$. The proof is complete.
\end{proof}

\medskip

\begin{proof}[Proof of Lemma~\ref{lem:exp_decay}]  
Consider $R > 0$ large enough to ensure 
$$ 
\frac{\theta}{2} \leq {\cal V}(\br) + \theta \leq \frac{3\theta}{2} \quad
\mbox{a.e. on } B_R^c := \RR^3 \setminus \overline B_R.
$$
It is straightforward to see that $u$ is the unique solution in
$H^1(B_R^c)$ to the elliptic boundary problem 
\begin{eqnarray*}
\left \{
\begin{aligned}
 -\mathrm{div}(A \nabla v) + {\cal V} v + \theta v &= 0 \quad \mbox{in }
 \; B_R^c \\ 
v &= u \quad \mbox{on } \partial B_R.
\end{aligned}
\right.
\end{eqnarray*}
Let $\gamma > 0$, $\tilde{u} = u \exp^{-\gamma(|\cdot|-R)}$  and $w = u
- \tilde{u}$. The function $w$ is in $H^1(\mathbb{R}^3)$ and is the
unique solution in $H^1(B_R^c)$ to 
\begin{eqnarray} \label{exp}
\left \{
\begin{aligned}
-\mathrm{div}(A \nabla w) + {\cal V} w + \theta w &= \mathrm{div}(A \nabla \tilde{u}) - {\cal V} \tilde{u} - \theta \tilde{u} \quad \mathrm{in} \; B_R^c \\
w &= 0 \quad \mbox{on } \partial B_R.
\end{aligned}
\right.
\end{eqnarray}
Let us now introduce the weighted Sobolev space $W_{0}^{\gamma}(B_R^c)$
defined by 
$$
W_{0}^{\gamma}(B_R^c) = \left\{ v \in H^{1}_{0}(B_R^c) \; | \; e^{\gamma
    |\cdot|} v \in H^1(B_R^c) \right\}
$$
endowed with the inner product 
$$  
(v,w)_{W_{0}^{\gamma}(B_R^c)} = \int_{B_R^c} 
e^{\gamma |\br|} (v(\br) w(\br)+\nabla v(\br) \cdot \nabla w(\br)) \, d\br.
$$
Multiplying (\ref{exp}) by $\phi e^{2\gamma |\cdot|}$ with $\phi \in
\mathcal{D}(B_R^c)$ and integrating by parts, we obtain 
$$ 
\int_{B_R^c} A \nabla w \cdot \nabla (\phi e^{2\gamma |\br|}) +
\int_{B_R^c} ({\cal V}+ \theta)w \phi e^{2\gamma |\br|} = - \int_{B_R^c} A
\nabla \tilde{u} \cdot \nabla (\phi e^{2\gamma |\br|}) - \int_{B_R^c}
({\cal V}+ \theta)\tilde{u} \phi e^{2\gamma |\br|} 
$$
and then
\begin{eqnarray} \label{var}
\begin{aligned}
& \int_{B_R^c} A e^{\gamma |\br|} \nabla w \cdot e^{\gamma |\br|} \nabla \phi +  2 \gamma \int_{B_R^c}A e^{\gamma |\br|} \nabla w \cdot \frac{\br}{|\br|} e^{\gamma |\br|} \phi +  \int_{B_R^c} ({\cal V}+ \theta) e^{\gamma |\br|} w e^{\gamma |\br|} \phi \\
 = & -\int_{B_R^c} A e^{\gamma |\br|} \nabla \tilde{u} \cdot e^{\gamma |\br|} \nabla \phi - 2 \gamma \int_{B_R^c}A e^{\gamma |\br|} \nabla \tilde{u} \cdot \frac{\br}{|\br|} e^{\gamma |\br|} \phi- \int_{B_R^c} ({\cal V}+ \theta) e^{\gamma |\br|} \tilde{u} e^{\gamma |\br|} \phi .
\end{aligned}
\end{eqnarray}
Due to the definitions of $W_{0}^{\gamma}(B_R^c)$ and $\tilde{u}$,
(\ref{var}) actually holds for $(w, \phi) \in W^{\gamma}_{0}(B_R^c) \times
W^{\gamma}_{0}(B_R^c)$, and it is straightforward to see that (\ref{var})
is a variational formulation equivalent to (\ref{exp}).\\ 
It is also easy to check that the right-hand-side in (\ref{var}) is a
continuous form on $W^{\gamma}_{0}(B_R^c)$, so that we only have to prove the
coercivity of the bilinear form in the left-hand-side of (\ref{var}) to be
able to apply Lax-Milgram lemma. We have for $v \in   W^{\gamma}_{0}(B_R^c)$
\begin{eqnarray*}
\begin{aligned}
& \int_{B_R^c} A e^{\gamma |\br|} \nabla v \cdot e^{\gamma |\br|} \nabla v + 2 \gamma \int_{B_R^c}A e^{\gamma |\br|} \nabla v \cdot \frac{\br}{|\br|} e^{\gamma |\br|} v +  \int_{B_R^c} ({\cal V}+ \theta) e^{\gamma |\br|} v e^{\gamma |\br|} v \\
& \geq \lambda \int_{B_R^c} \left|e^{\gamma |\br|} \nabla v \right|² - 2 \Lambda \gamma \int_{B_R^c} \left|e^{\gamma |\br|} \nabla v \right| \left|e^{\gamma |\br|} v \right| + \frac{\theta}{2} \int_{B_R^c} \left|e^{\gamma |\br|} v \right|²\\
& \geq \lambda \left\|e^{\gamma |\br|} \nabla v \right\|^2_{L²(B_R^c)} - 2 \Lambda  \gamma \left\|e^{\gamma |\br|} \nabla v \right\|_{L²(B_R^c)} \left\|e^{\gamma |\br|} v \right\|_{L²(B_R^c)} + \frac{\theta}{2} \left\|e^{\gamma |\br|} v \right\|^2_{L²(B_R^c)} \\
& \geq (\lambda - \Lambda \gamma) \left\|e^{\gamma |\br|} \nabla v
\right\|^2_{L²(B_R^c)} +  (\frac{\theta}{2} - \Lambda \gamma)
\left\|e^{\gamma |\br|} v \right\|^2_{L²(B_R^c)} .
\end{aligned}
\end{eqnarray*}
Thus the bilinear form is clearly coercive if $\gamma <
\min(\frac{\lambda}{\Lambda}, \frac{\theta}{2\Lambda})$, and there is a
unique $w$ solution of (\ref{exp}) in $W_{0}^{\gamma}(B_R^c)$ for such a
$\gamma$. Now since $u = w + \tilde{u}$, it is clear that $e^{\gamma
  |\cdot|} u \in H¹(B_R^c)$, and then $e^{\gamma |\cdot|} u \in
H¹(\mathbb{R}^3)$.
\end{proof}

\section*{Acknowledgements.} The authors are grateful to C. Le Bris and
M. Lewin for helpful discussions. This work was completed while E.C. was
visiting the Applied Mathematics Division at Brown University.

\end{document}